\documentclass[a4paper, 12pt]{article}

\usepackage{amsthm,stmaryrd}
\usepackage{amsmath}
\usepackage{amssymb}
\usepackage{latexsym}
\usepackage{authblk}
\usepackage[normalem]{ulem}
\usepackage{tcolorbox}

\usepackage{hyperref} %[backref]
\usepackage[margin=20mm]{geometry}

\newtheorem{theorem}{Theorem}[section]
\newtheorem{proposition}[theorem]{Proposition}
\newtheorem{lemma}[theorem]{Lemma}
\newtheorem{corollary}[theorem]{Corollary}
\newtheorem{remark}[theorem]{Remark}
\newtheorem{definition}[theorem]{Definition}
\newtheorem{example}[theorem]{Example}
\newtheorem{assumption}[theorem]{Assumption}

\def\A{{\mathcal{A}}}
\def\H{{\mathcal{H}}}
\def\F{{\mathcal{F}}}

\newcommand{\bW}{\mathbf{W}}
\newcommand{\bw}{\mathbf{w}}
\newcommand{\bC}{\mathbf{C}}
\newcommand{\bL}{\mathbf{L}}
\newcommand{\bK}{\mathbf{K}}

\newcommand{\bff}{\mathbf{f}}

\newcommand{\bQ}{\mathbf{Q}}

\newcommand{\bZ}{\mathbf{Z}}

\newcommand{\bv}{\mathbf{v}}

\def\A{{\mathcal{A}}}
\def\H{{\mathcal{H}}}
\def\F{{\mathcal{F}}}

\begin{document}
\title{A general framework for pricing and hedging under local viability\footnote{The second author gratefully acknowledges the support of the National Research, Development and Innovation Office (NKFIH) through grant K 143529 and also within the framework of the Thematic Excellence Program 2021 (National Research subprogramme ``Artificial intelligence, large networks, data security: mathematical foundation and applications''). We thank the referees for pointing out a mistake regarding the definition of the superhedging price in the previous version of the paper. }}

\author[]{Huy N. Chau}
\author[2]{Mikl\'os R\'asonyi}
\affil[1]{Department of Mathematics, University of Manchester, UK}
\affil[2]{HUN-REN Alfr\'ed R\'enyi Institute of Mathematics and E\"otv\"os Lor\'and University, Budapest, Hungary}

\maketitle

\begin{abstract}
In this paper, a new approach for solving the problems of pricing and hedging derivatives 
is introduced in a general frictionless market setting. The method is applicable even in cases where an equivalent 
local martingale measure fails to exist. Our main results include a new superhedging duality for American options when wealth processes can be negative and trading strategies are subject to a cone constraint.  
This answers one of the questions raised by Fernholz, Karatzas and Kardaras in \cite{fkk2005}.
\end{abstract}

\section{Introduction}
In Mathematical Finance, the Fundamental Theorem of Asset Pricing (FTAP) 
states that, rougly speaking, no-arbitrage is tantamount to the existence of equivalent local martingale measures. 
Some earliest results include \cite{ross77}, \cite{dr}, \cite{kreps81}. Discrete-time versions were obtained 
in \cite{dmw}, \cite{harrison1979martingales}. Versions of FTAP for general semimartingales in continuous time 
were given in \cite{delbaen1994general}, \cite{ds1998}. 
The authors introduced therein the condition No Free Lunch with Vanishing Risk (NFLVR) and  
proved that it is equivalent to the existence of local martingale measures (for locally bounded semimartingales). 
We refer to \cite{delbaen2006mathematics} for a comprehensive treatment of no-arbitrage pricing theory. 

The NFLVR condition provides a typical framework where the problems of pricing, hedging 
or portfolio optimization can be formulated. 
However, requiring absence of free lunches with vanishing risk can sometimes be restrictive:
in a model with the three-dimensional Bessel process (see \cite{ds95}) or in the Stochastic 
Portfolio Theory of \cite{fk}, the NFLVR condition fails and arbitrage opportunities may arise. 
Several approaches have been proposed for such models. One could apply the benchmark approach: 
a pricing theory under physical measure, developed in \cite{platen2006}, see also \cite{platenheath2006}. 
Alternatively, one could try to develop a theory based on a weaker no-arbitrage condition. 
No Unbounded Profit with Bounded Risk (NUPBR), introduced in \cite{kk}, is such an alternative. 
The authors in \cite{kk} showed that the problems of pricing and utility maximization can still be solved 
under NUPBR. In \cite{k12}, the author proved the equivalence of NUPBR and the existence of a 
local martingale deflator for one dimensional asset price processes. The paper 
\cite{takaoka} extended this result to a multidimensional semimartingale setting. 
The connection of NUPBR and numerair\'e portfolios was discussed in \cite{cdm15}. 
In a diffusion setting, \cite{ruf} developed pricing equations and hedging formulas under NUPBR. 
See also \cite{claudio} for other concepts of no-arbitrage. 

An American option is a derivative which gives holders the right to exercise the option at any
instant (stopping time) until a given future date. 
In the pioneering work \cite{MacKean}, the author transformed the pricing problem for 
an American option into a free boundary problem. The option price was computed explicitly 
up to the optimal stopping boundary. The idea was investigated further by \cite{moerbeke76}. 
Financial hedging arguments for American claims were introduced later in 
\cite{Bensoussan84}, \cite{karatzas1988}, and \cite{karatzas1989} for diffusion settings. 
For the optimal stopping and free boundary problems, we refer to the book \cite{ps2006}. In another 
line of research, techniques with variational inequalities, BSDEs were applied to compute the prices of 
American options in \cite{bl82} and \cite{jll}, \cite{Karoui1997}, \cite{grigorova2017}, among other papers. Summary of the most essential 
results on the pricing of American options were discussed in \cite{myneni}.

Most of the literature on American options assumes that an equivalent local martingale measure (ELMM) exists. Under arbitrage, there are no ELMMs, and the problems of pricing and hedging are little studied, with the exceptions of \cite{bkx},  \cite{kk2021}.  
As observed by C. Kardaras, in the absence of an ELMM it is not optimal to exercise an American call option 
(written on a non-dividend-paying stock) only at the maturity date. 
Can one then characterize, or compute, the optimal exercise time? The question 
was answered in \cite{bkx} where the optimal stopping time to exercise American call options was derived. 
The solution is meaningful for option holders. However, from the sellers' perspective, 
this is not enough and the following questions arise naturally:\\

Q1: How can we provide a hedging argument for the seller?\\ 

More generally, in \cite{fkk2005}, see also in  \cite{fk}, the authors posed the following problem:\\

Q2: Develop a theory for pricing American contingent claims without ELMMs.\\

In the first part of this paper, we introduce a new framework to study the problems of pricing and hedging for American options in the absence of ELMMs, or, 
more precisely, when the condition NFLVR holds only locally (that is, up to a sequence of stopping times) and not globally. 
The framework covers many interesting situations where arbitrage opportunities exist, as explained above. Unlike previous studies, in our setting option sellers \emph{continue} trading after the exercise time chosen by buyers to lower hedging price. The mathematical background is, as usual, the closedness of the set of hedgeable claims in appropriate topological spaces.
We employ an idea with product spaces which is rooted in \cite{chau2020} for discrete time models, 
in \cite{cfr} for super-replication under model uncertainty and transaction costs. 
In addition, our techniques require the whole machinery developed in \cite{delbaen1994general}.  As far as we know, the current paper is the first to study local viability in general frictionless market 
settings when portfolios are allowed to be negative and strategies take values in a convex cone. 

We obtain new pricing systems which are more complex than the local martingale deflators developed 
in \cite{kk}, \cite{k12}, \cite{takaoka}, \cite{cdm15}, just to mention a few. 
Importantly, under the local NFLVR condition, we are able to establish the existence of superhedging strategies for American options by using the new framework and superhedging duality by using the new pricing systems. A superhedging duality result for American option were given in \cite{kk2021} for \emph{continuous} price processes using tools in stochastic analysis such as optional decomposition (we believe that these results of \cite{kk2021} could be extended to c\`adl\`ag price processes by using the general optional decomposition given in \cite{sy1998}). Nevertheless, the duality of \cite{kk2021} is less practical than ours since sellers in their framework have to stop trading after the exercise time. In addition, our arguments purely rely on  analysis, do not involve change of num\'eraire techniques or optional decomposition, and 
thus are applicable not only to negative portfolios but also when other factors, such as dividends, liquidity costs, 
etc., are taken into account.   
Let us recall the paper \cite{Kardaras} which explains that the notion of num\'eraire portfolio is 
only available for nonnegative processes, and asked for different, appropriate tools for working with negative wealth 
processes, see Remark 1.12 therein. Such requests are particularly useful for utility maximization problems, where the initial capital is fixed. 

In the second part of the paper, we study the pricing and hedging problems for European options.  It is shown that the new superhedging duality is stated by using simpler versions of the new pricing systems as European options are special cases of American options. In addition, when the option maturity is exactly the trading horizon, the superhedging price can be computed by using only the concept of Equivalent Local Martingale Deflator, see \cite{k12}. It is noted that the new duality works with negative portfolios,  trading constraints, and the approach is suitable for more complex situations.  

To sum up, we provide a complete answer for the questions Q1 and Q2 (recalled above). It is worth emphasizing that even when the market satisfies the global NFLVR condition,  our approach for American options is \emph{new} and may be used in future work.  

The paper is organised as follows. Section \ref{sec:model} introduces the model. 
The new approach and results for American options are discussed in Section \ref{sec:american}. Section \ref{sec:eur} discusses results for European options. Proofs are given in 
Section \ref{sec:proofs}. Section \ref{sec:app} collects useful auxiliary material. 

\textit{Notations}. Let $I$ be some index set and $X_{i}$, $i\in I$ be sets.
In the product space  $\mathbf{X}=\prod_{i \in I} X_i$, a vector $(f^i)_{i \in I}$ will be 
denoted by $\mathbf{f}$. If there are orderings $\geq_{i}$ given on each $X_{i}$ then
we write $\mathbf{f} \ge \mathbf{g}$ if $f^i \geq_{i} g^i$ for all $i \in I$. 
If $1\in X_{i}$ for all $i$ then $\mathbf{1}$ denotes the vector with all coordinates equal to $1$ 
and $\mathbf{1}^{i}$ denotes the vector with coordinate $i$ equal to $1$ and the other coordinates zero.  
Similarly, when $0\in X_{i}$, $\mathbf{0}$ denotes the vector all of whose coordinates equal $0$. 
We denote by $L^0(\mathcal{F}, P)$ the vector space of (equivalence classes of) random variables on $(\Omega, \mathcal{F}, P)$, equipped with the metric $d_0(X,Y) = E[|X-Y|\wedge 1]$ for any $X,Y \in L^0(\mathcal{F}, P)$.
As usual, $L^{p}(\mathcal{F}, P), p \in [1, \infty]$ is the space of $p$-integrable (resp. bounded)
random variables equipped with the 
standard $\|\cdot\|_{p}$ norm.

\section{The model}\label{sec:model}
%Let $T>0$ be a fixed time horizon. 
Let $(\Omega, \mathcal{F}, (\mathcal{F}_t)_{t \ge 0}, P)$ be a filtered probability space, where the filtration is assumed to be right-continuous and $\mathcal{F}_0$ contains all
null sets of $\mathcal{F}$. We can assume further that $\mathcal{F}_{\infty}= \bigvee_{t \in \mathbb{R}_+}\mathcal{F}_t$. Consider a financial market model with $1\le d \in \mathbb{N}$ risky assets $S = (\mathcal{S}^i)_{1\le i \le d}$ and one risk-free asset whose price is assumed to be one at all times. Assume that $S$ is an $\mathbb{R}^d$-valued,  adapted and locally 
bounded 
semimartingale. For a given $\mathbb{R}^d$ -valued semimartingale $S$, we write $\mathcal{L}(S)$ for the space of $\mathbb{R}^d$-valued,
$S$-integrable, predictable processes $H = (H^i)_{1\le i \le d}$ and $W_t(H) = H \cdot S_t = \int_0^t{H_udS_u}$ for the corresponding vector stochastic integral, which is a one dimensional process, see \cite{sc2002}.

Let $\mathfrak{C}$ be a closed convex polyhedral cone of $\mathbb{R}^d$, representing trading constraints.
\begin{definition}\label{defi:admi_appro}
	A trading strategy $H$ is $x$-admissible if $H_0 = 0, H_t(\omega) \in \mathfrak{C}, a.s.$ for all $t \ge 0$, and if 
	\begin{equation}\label{}
		H \cdot S_t \ge - x,\mbox{ a.s.,  for all }t \ge 0,
	\end{equation}
	Denote by $\mathcal{A}_{x}$ the set of $x$-admissible strategies, and $\mathcal{A}= \bigcup_{x> 0}\mathcal{A}_{x}$.
\end{definition}
\begin{example} Some examples of cone constraints are given below:
	\begin{itemize}
		\item unconstrained case: $\mathfrak{C}^{u} := \mathbb{R}^d$.
		\item no short-sale constraint on the first $1 \le n \le d$ assets:
		$$\mathfrak{C}^{s} := \{ (h^i)_{1\le i \le d} \in \mathbb{R}^d: h^i \ge 0 \text{ for } 1 \le i \le n\}.$$
	\end{itemize} 
	Other constraints can be found in \cite{napp}.
\end{example}
\begin{definition}\label{defi:local_seq}
	A nondecreasing sequence of stopping times $\mathcal{T}= \{T_k, k \in \mathbb{N}\}$ such that 
\begin{equation}\label{defi:local}
\lim_{k \to \infty}P(T_k =  \infty)= 1
\end{equation} is called a localizing sequence. 
\end{definition}
Since $S$ is locally bounded, we may assume that $S^{T_k}_t:=S_{t \wedge T_k}$ is bounded without loss of generality.  We define 
\begin{eqnarray*}
K_0 &=& \{ (H \cdot S)_{\infty}: H \in \A \text{ and } (H \cdot S)_{\infty}:= \lim_{t \to \infty} (H \cdot S)_t  \text{ exists } a.s.\},
\end{eqnarray*}
and also $C = C_0 \cap L^{\infty}$ where $C_0 = K_0 - L^0_+$. We define the corresponding local versions
\begin{eqnarray*}
	K^k_0 &=& \{ (H \cdot S^{T_k})_{\infty}: H \in \A \text{ and } (H \cdot S^{T_k})_{\infty}:= \lim_{t \to \infty} (H \cdot S^{T_k})_t  \text{ exists } a.s.\},
\end{eqnarray*}
and $C^k = C^k_0 \cap L^{\infty}$ where $C^k_0 = K^k_0 - L^0_+$. We denote by  $\overline{C}^{\|\|_{\infty}}$ and $\overline{C^k}^{\|\|_{\infty}}$ the closures of $C$ and $ C^k$ in the sup norm topology of $L^{\infty}$, respectively. The no-arbitrage condition below is classical. 

\begin{definition} We say that the market satisfies the condition
	\begin{itemize}
		\item[(i)] no arbitrage (NA$_{\mathfrak{C}}$) if $C \cap L^{\infty}_+ = \{0\}$;
		\item[(ii)] local no arbitrage (local NA$_{\mathfrak{C}}$) w.r.t $\mathcal{T}$ if $C^k \cap L^{\infty}_+ = \{0\}$ for all $k \in \mathbb{N}$. 		
	\end{itemize}
\end{definition}
The No Free Lunch with Vanishing Risk condition is recalled from \cite{delbaen1994general} (when $\mathfrak{C} = \mathfrak{C}^u$),  \cite{kk} and its local version is also introduced below. %{\color{red}this was not defined, I think it should be states that this means (NFLVR$\mathbb{R}^d$)}
\begin{definition}\label{defi:hflvr} The market satisfies the condition
	\begin{itemize}
		\item[(i)] no free lunch with vanishing risk (NFLVR$_{\mathfrak{C}}$) if $\overline{C}^{\|\|_{\infty}} \cap L^{\infty}_+ = \{0\}$;
		\item[(ii)] local no free lunch with vanishing risk (local NFLVR$_{\mathfrak{C}}$) w.r.t $\mathcal{T}$ if for all $k \in \mathbb{N}$, it holds that  $\overline{C^k}^{\|\|_{\infty}} \cap L^{\infty}_+ = \{0\}$.
	\end{itemize}
\end{definition}
The concept of No Unbounded Profit with Bounded Risk (NUPBR$_{\mathfrak{C}}$), see \cite{kk}, \cite{k12}, \cite{ccfm},  is recalled below. It was called BK in \cite{kabanov}. Here, we adopt the version from \cite{ccfm} for infinite horizon settings. We define 
\begin{equation*}
	\mathcal{X}:= \{ W \ge 0: W_t = 1 + (H \cdot S)_t, \ t \ge 0\}. 
\end{equation*}
\begin{definition}
 The market satisfies the condition NUPBR$_{\mathfrak{C}}$ if $$\mathcal{X}_T:=\{ W_T: W \in \mathcal{X} \}$$ is bounded in $L^0$ for every $T > 0$.
\end{definition}

For the case without trading constraints, a general version of the first FTAP was given in Corollary 1.2 of \cite{delbaen1994general}. For a comprehensive theory of arbitrage, we refer to the book \cite{delbaen2006mathematics}.
\begin{theorem}\label{thm:classical} For a locally bounded semimartingale $S$, the condition NFLVR$_{\mathfrak{C}}$ when $\mathfrak{C} = \mathfrak{C}^u$ is equivalent to the existence of a
	probability $Q\sim P$ such that $S$ is a local $Q$-martingale.
\end{theorem}
The global NFLVR$_{\mathfrak{C}}$ condition is a special case of the local NFLVR$_{\mathfrak{C}}$ condition when the localizing sequence contains only infinity. When  $\mathfrak{C}$ is predictable closed convex cone, it is known that NFLVR$_{\mathfrak{C}}$ $=$ NUPBR$_{\mathfrak{C}}$ $+$ NA$_{\mathfrak{C}}$, see Proposition 4.2 of \cite{kk}. Similar results in frictionless settings can be found in Corollary 3.8 of \cite{delbaen1994general}, Lemma 2.2 of \cite{kabanov}. Furthermore, local NUPBR$_{\mathfrak{C}}$ is equivalent to global NUPBR$_{\mathfrak{C}}$, while the local NA$_{\mathfrak{C}}$ condition does not imply the global NA$_{\mathfrak{C}}$ one,  see \cite{ruf}, Example 4.6 of \cite{kk}. Therefore, it is possible to have arbitrage opportunities under the local NFLVR$_{\mathfrak{C}}$ condition. 

Equivalent local martingale deflators were introduced in \cite{k12}, see also \cite{ccfm} for the infinite horizon settings, which play the same roles as equivalent martingale measures. 
\begin{definition}
Let $\mathfrak{C} = \mathfrak{C}^u$. An equivalent local martingale deflator (ELMD) is a nonnegative local martingale  $Y$ such that $Y_0 = 1$ and $YW$ is a local martingale for all $W \in \mathcal{X}$.
\end{definition}
When $\mathfrak{C} = \mathfrak{C}^u$ and $S$ is a one dimensional semimartingale, \cite{k12} proved the equivalence between NUPBR$_{\mathfrak{C}}$ and the existence of an ELMD. By introducing the concept of strict sigma-martingale
density, \cite{takaoka} generalized the result for  finite-dimensional semimartingale settings using a change of num\'eraire technique. \cite{choulli} proved the equivalence between NUPBR$_{\mathfrak{C}}$ and the existence of a strict sigma-martingale density for continuous semimartingales by using only stochastic calculus. It is worth noting that all of the previous studies focus on nonnegative portfolios. 

\section{Pricing and hedging American options}\label{sec:american}
Define $\mathbb{T}:=\{ \tau: \tau \text{ is a stopping time taking values in } [0,\infty] \}.$ Let $\Phi: \Omega \times [0, \infty] \to \mathbb{R}$ be measurable. An American option $\Phi$ gives its holders the payoff $\Phi_{\tau}1_{\tau < \infty}$ when exercised at time $\tau\in\mathbb{T}$. On the event $\{\tau = \infty\}$, the option is not exercised, hence the payoff is zero. A European option is a special case when $\mathbb{T} = \{ T^{\sharp} \}$ where $T^{\sharp} > 0$ is the maturity of the European option. A Bermudan option is also a special case where the set of possible exercise times is finite. Let $\mathcal{T}= \{T_k, k \in \mathbb{N}\}$ be a localizing sequence. The following conditions will be imposed.

\begin{assumption}\label{assumption_main}
	\begin{itemize}
		\item[(i)] $S$ satisfies the local NFLVR$_{\mathfrak{C}}$ condition w.r.t. $\mathcal{T}$.  
		\item[(ii)] $\Phi$ is a.s. continuous from the right: for almost every $\omega\in\Omega$, for all $t \in [0, \infty)$, if $t<t_{n}\in [0,\infty)$, $n\in\mathbb{N}$ converges to $t$ then $$\lim_{n \to \infty} \Phi_{t_n}(\omega) = \Phi_t(\omega).$$
		We assume in addition that $\Phi_{\infty}=0$.
	\end{itemize}
\end{assumption}
Under Assumption \ref{assumption_main} (ii), $\Phi$ is optional. Assumption \ref{assumption_main} is very general and includes most interesting models.  Although the local NFLVR$_{\mathfrak{C}}$ and the NUPBR$_{\mathfrak{C}}$ conditions are equivalent, we choose the formulation with the local NFLVR$_{\mathfrak{C}}$ to work with negative portfolios. Furthermore, we are able to use the geometric representation in  \eqref{eq:nflvr1} and the standard separation argument. In this way, the local NFLVR$_{\mathfrak{C}}$ condition  can be seen as a unified framework for the global NFLVR$_{\mathfrak{C}}$ and the NUPBR$_{\mathfrak{C}}$ conditions. 

In classical settings, the superhedging price of the American option $\Phi$ is defined by
\begin{eqnarray}\label{eq:price_American}
	\inf \left\lbrace  z: \exists H \in \mathcal{A} \text{ such that } z + H \cdot S_{\tau} \ge \Phi_{\tau}, \ a.s., \text{ for any } \tau \in \mathbb{T} \right\rbrace. 	
\end{eqnarray}
For hedging, the option seller has to find \emph{one} strategy $H$ such that the corresponding wealth process dominates the American payoff at any stopping time. In the presence of arbitrage opportunities, the seller may prefer to continue trading after the stopping time $\tau$ chosen by the option buyer and exploit such riskless opportunities to reduce hedging prices. Notably, the trading strategy after the exercise time $\tau$ may be different from the strategy before $\tau$. For this modelling purpose, we need to consider a generalization of admissibility. Let $\tau \in \mathbb{T}$ be a stopping time and $v$ be an $\mathcal{F}_{\tau}$-measurable random variable. We define  
\begin{eqnarray*}
\mathcal{A}_x(\tau,v) &:=& \left\lbrace H \text{ such that}\  H_t1_{\tau < t} \in \mathcal{L}(S),   H_t(\omega)1_{\tau < t} \in \mathfrak{C} \right. \\
&& \qquad \left.  v + \int_{\tau}^t H_udS_u \ge -x, \ a.s. \forall \tau \le t \right\rbrace,
\end{eqnarray*}
and $\mathcal{A}(\tau,v) = \cup_{x \ge 0}\mathcal{A}_x(\tau,v)$. We define the set of $x$-generalized strategies by
\begin{equation}\label{eq:ge_stra1}
	\widetilde{\mathcal{A}}^{\mathbb{T}}_x = \left\lbrace  \widetilde{H} = \left( H, (H(\tau))_{\tau \in \mathbb{T}} \right), H \in \mathcal{A}_x \text{  and } H(\tau)1_{]]\tau,\infty[[} \in \mathcal{A}_x\left( \tau,\int_{0}^{\tau}H_udS_u\right) \right\rbrace ,
\end{equation}
and
\begin{equation}\label{eq:tilde_H}
	\widetilde{H}_t(\tau) := H_t1_{t \le \tau} + H_t(\tau)1_{t > \tau}, \  \forall  \tau \in \mathbb{T}.
\end{equation}
Define also $\widetilde{\mathcal{A}}^{\mathbb{T}} = \cup_{x >0} \widetilde{\mathcal{A}}^{\mathbb{T}}_x$. For each $\tau \in \mathbb{T}$, the strategy $\widetilde{H}_t(\tau)$ in \eqref{eq:tilde_H} is also predictable. For each strategy $\widetilde{H} = \left( H, (H(\tau))_{\tau \in \mathbb{T}} \right)$, $H$ represents the strategy that the seller uses before the exercise time $\tau$, and $H(\tau)1_{]]\tau,\infty[[}$ is the strategy used after time $\tau$. Note that a strategy $\widetilde{H}$, parametrized by $\tau$, is an infinite dimensional \emph{vector} of admissible trading strategies. For a strategy  $\widetilde{H}$, the corresponding terminal wealth is a vector $(\widetilde{H}(\tau)\cdot S_{\infty})_{\tau \in \mathbb{T}}$ if the limit vector exists. For example, under the condition local  NFLVR$_\mathfrak{C}$ w.r.t $\mathcal{T}$, Theorem 9.3.3 of \cite{delbaen2006mathematics}, implies that for any $\tau$, the random variable $\lim_{t \to \infty}(\widetilde{H}(\tau)\cdot S)_{t \wedge T_k}$ exists and is finite almost surely for any $T_k$. Therefore, the random variable $(\widetilde{H}(\tau)\cdot S)_{\infty}$ is well-defined, by using  \eqref{defi:local}. 

For superhedging purposes, we assume that $\Phi_t \ge -x, a.s.$ for all $t$ and the credit constraint is also $x$, e.g., wealth processes must be bounded from below by $-x$ at all time.  If the American option is exercised at $\tau < \infty$, the payoff $\Phi_{\tau}$ is delivered at time $\tau$ and the credit constraint is satisfied by a hedging strategy $H$ if
\begin{equation}\label{at_tau}
		z + \int_{0}^{\tau} H_udS_u - \Phi_{\tau} \ge - x, \ a.s.
\end{equation}
For $\tau < t$, we use a strategy $H(\tau)$ such that 
\begin{equation}\label{after_tau}
		z + \int_0^{\tau} H_udS_u - \Phi_{\tau} + \int_{\tau}^t H_u(\tau)dS_u \ge -x, \ a.s.
	\end{equation}
This condition requires that after delivering the payoff at time $\tau$, sellers fulfill the credit constraint. The terminal  wealth needs to be nonnegative
	\begin{equation}\label{at_T}
		z + \int_0^{\tau} H_udS_u - \Phi_{\tau} + \int_{\tau}^{\infty} H_u(\tau)dS_u \ge 0, \ a.s.
	\end{equation}
This approach partially covers the American payoff starting from time $\tau$, as in \eqref{at_tau} and \eqref{after_tau}. The superhedging procedure is fully complete at time infinity.   
This leads to the following definition of superhedging price, see similar concepts in \cite{bhz2015}, \cite{bz2017}, \cite{adot},  
\begin{eqnarray}\label{defi:superhedging_new}
	\pi^A_x(\Phi) &:=& \left\lbrace z \in \mathbb{R}: \exists \widetilde{H} \in \widetilde{\mathcal{A}}^{\mathbb{T}}\text{ s.t. } \forall \tau \in \mathbb{T}, t \in [0, \infty],   \right.  \nonumber \\
	&& \qquad \left. z + \int_{0}^{t}\widetilde{H}_u(\tau) dS_u - \Phi_{\tau}1_{\tau \le t} \ge -x1_{t < \infty}, \ a.s.  \  \right\rbrace.
\end{eqnarray}
Here, the hedging portfolio could be negative, and the credit constraint $x$ is satisfied at all time. 
\begin{remark}
The inequality constraints in \eqref{at_tau} and \eqref{after_tau} should be interpreted as superhedging constraints rather than admissibility constraints. Admissibility constraints are used to rule out doubling strategies and independent from option payoffs.  If we understand \eqref{after_tau} as an admissibility constraint, then we will ask $H(\tau)1_{]]\tau,\infty[[} \in \A_{x}\left( \tau, 	z + \int_0^{\tau} H_udS_u - \Phi_{\tau}   \right).$ This condition depends on the payoff $\Phi_{\tau}$. Therefore, this is not the right formulation for the no arbitrage property. 
\end{remark} 
\begin{remark}
In discrete-time settings, the papers \cite{adot}, 	\cite{bhz2015} and 	\cite{bz2017} allow trading after the exercise time $\tau$ and the payoff can be delivered at the final time as convention. This is equivalent to delivering the payoff at its exercise time by taking a loan, and then the sellers continue trading in such a way that their final wealth is nonnegative. This is possible because there is no admissibility constraint in discrete-time settings. In the present continuous time setting, this convention does not work because of the admissibility and credit constraints. At time $\tau$, delivering the possibly unbounded payoff has a  significant impact\footnote{We thank the referees for pointing out this issue.} to the admissibility for the hedging portfolio at time $\tau$ and after time $\tau$.
\end{remark}

To derive superhedging dualities for American options, we transform the superhedging problem for American options into the one for the corresponding European options in appropriate product spaces. Since our setting allows arbitrage opportunities, the superhedging problem has solutions if trading strategies are bounded from below, see also \cite{ccf} for similar ideas for utility maximization. Here, we assume that the credit constraint $x$ is fixed. 

Let $t_i, i \in \mathbb{N}$ be an enumeration of $(\mathbb{Q} \cap [0,\infty)) \cup \{\infty\}$ with
$t_0 = \infty$. Define by $\Theta$ the set containing all such $t_i$,
\begin{equation}\label{defi:theta}
	\Theta =\{  t_i, i\in \mathbb{N} \}.
\end{equation}
Similar to \eqref{eq:ge_stra1}, we define $ \widetilde{\mathcal{A}}^{\Theta} = \cup_{x >0} \widetilde{\mathcal{A}}^{\Theta}_x$ where
\begin{equation}
	\widetilde{\mathcal{A}}^{\Theta}_x = \left\lbrace  \widetilde{H} = \left( H, (H(\theta))_{\theta \in \Theta} \right), H \in \mathcal{A}_x \text{  and } H(\theta)1_{]]\theta,\infty[[} \in \mathcal{A}_x\left( \theta,\int_{0}^{\theta}H_udS_u\right) \right\rbrace.
\end{equation}
The following proposition shows that it is enough to superhedge $\Phi$ at rational times. 
\begin{proposition}\label{pro:reduce}  Let $z$ be a real number and assume $\Phi_t \ge -x, \ a.s.$ The following are equivalent:
\begin{itemize}
		\item[(i)] There exists $\widetilde{H} \in \widetilde{\mathcal{A}}^{\mathbb{T}}$ such that $\forall \tau \in \mathbb{T}, t \in [0,\infty]$, 
		\begin{equation}\label{eq:hedge}
			z + \int_{0}^{t}\widetilde{H}_u(\tau) dS_u - \Phi_{\tau}1_{\tau \le t} \ge -x1_{t < \infty}, \ a.s., 
		\end{equation}
		\item[(ii)] There exists $\widetilde{G} \in \widetilde{\mathcal{A}}^{\Theta}$ such that $\forall
		k \in \mathbb{N}, \theta, \zeta \in \Theta,$
		\begin{equation}\label{eq:hedge2_2}
			z + \int_0^{\zeta \wedge T_k} \widetilde{G}_u(\theta) dS_u - \Phi_{\theta}1_{\theta\le  \zeta \wedge T_k}  \ge  - x 1_{\zeta \wedge T_k < \infty}, \ a.s.
		\end{equation}
	\end{itemize} 
\end{proposition} 

Motivated by \eqref{eq:hedge2_2}, for a generalized strategy $\widetilde{H}$, we define the corresponding stopped processes for $(k,\theta, \zeta) \in \mathbb{N} \times \Theta \times \Theta, t \in [0,\infty)$,
\begin{equation}\label{eq:stopped_wealth}
	\widetilde{W}^{k,\theta, \zeta}_t(\widetilde{H}) =\int_0^{t \wedge T_k \wedge \zeta}  \widetilde{H}_u(\theta)dS_u=\int_0^{t} \widetilde{H}_u(\theta)dS^{T_k \wedge \zeta}_u.
\end{equation}
In the following discussions, we will work with different product spaces.  Let $\emptyset \ne D \times \Gamma \times \Sigma \subset \mathbb{N} \times \Theta \times \Theta$. For $t \in [0, \infty]$, denote by $$\bL^{\infty, D \times \Gamma \times \Sigma}_t = \prod_{k \in D, \theta \in \Gamma, \zeta \in \Sigma} \left( L^{\infty}(\mathcal{F}_{t \wedge  T_k \wedge \zeta},P), w^* \right)$$ with the corresponding product topology when $L^{\infty}(\mathcal{F}_{t \wedge T_k\wedge \zeta},P)$ is equipped with the weak star topology $w^*=\sigma(L^{\infty}(\mathcal{F}_{t \wedge T_k\wedge \zeta},P),L^1(\mathcal{F}_{t \wedge T_k\wedge \zeta},P))$. Such product topology is called the $\bw^*$-topology. Furthermore, we define 
\begin{eqnarray}
	\bL^{\infty, D \times \Gamma \times \Sigma, b}_t &:=& \left\{  \bff \in \bL^{\infty, D\times \Gamma \times \Sigma}_t: \exists a \in \mathbb{R}_+  \text{ such that } \right. \nonumber\\
	&& \left. -a \le f^{k,\theta, \zeta} \le a, \text{ almost surely } \forall k \in D, \theta \in \Gamma, \zeta \in \Sigma \right\}. \label{defi:subspace}
\end{eqnarray}
\begin{remark}

The boundedness from below and above in \eqref{defi:subspace} make $\bL^{\infty, D\times \Gamma \times \Sigma, b}_t$ a subspace. We note that $\bL^{\infty, D\times \Gamma \times \Sigma, b}_t = \bL^{\infty, D\times \Gamma \times \Sigma}_t$ when $D\times \Gamma \times \Sigma$ is finite. It is crucial to work with the subspace $\bL^{\infty, D\times \Gamma \times \Sigma, b}_t$ to establish Fatou-closedness of certain sets, see Proposition \ref{pro:countable_Fatou} later on. Under the global NFLVR$_{\mathfrak{C}}$ condition, there is no need to restrict to these subspaces. It is clear that  $\bL^{\infty, D\times \Gamma \times \Sigma, b}_t$ is also a locally convex topological space, when equipped with the induced topology of $\bL^{\infty, D\times \Gamma \times \Sigma}_t$. Other product spaces, e.g., $\bL^{0,D\times \Gamma \times \Sigma,b}_t$, are defined similarly, see also \cite{cfr} for product spaces and their duals. The product space 
$\bL^{\infty, \mathbb{N}\times \Theta \times \Theta}_t$ 
admits the predual $\bigoplus_{(k,\theta, \zeta) \in \mathbb{N}\times \Theta \times \Theta} L^{1}(\mathcal{F}_{t \wedge T_k \wedge \zeta},P)$, which is not a Fréchet space, but an LF space, see Appendix \ref{sec:app}. Therefore we cannot apply the Krein-Smulian theorem (see in \cite{schachermayer1994}) to $\bL^{\infty, \mathbb{N}\times \Theta \times \Theta}_t$. Similarly, it is not easy to find the predual of $\bL^{\infty,\mathbb{N}\times \Theta \times \Theta,b}_t$ (see Theorem 8.12.1 of \cite{narici}) and check if the predual is a Fréchet space, in order to use the Krein-Smulian theorem. It is also emphasized that in the proof of Theorem \ref{thm:classical}, the Krein-Smulian theorem plays a crucial role while in our paper it can't be used and we use a compactness argument instead.
\end{remark}
For $D \times \Gamma \times \Sigma \subset \mathbb{N} \times \Theta \times \Theta$, we set
\begin{eqnarray}
	\widetilde{\bK}^{D \times \Gamma \times \Sigma}_x &=& \left\lbrace \left( \widetilde{W}^{k,\theta, \zeta}_{\infty}(\widetilde{H}) \right)_{k \in D,\theta \in \Gamma, \zeta \in \Sigma}: \widetilde{H} \in \widetilde{\mathcal{A}}^{\Theta}_x, \right.  \nonumber \\
	&&\left.  \qquad \text{ and } \left( \widetilde{W}^{k,\theta, \zeta}_{\infty}(\widetilde{H}) \right)_{k \in D,\theta \in \Gamma, \zeta \in \Sigma} \text{exists} \right\rbrace,  x > 0, 
	\label{eq:K_x}\\
\widetilde{\bK}^{D \times \Gamma \times \Sigma}_0 &=& \left\lbrace \left( \widetilde{W}^{k,\theta, \zeta}_{\infty}(\widetilde{H}) \right)_{k \in D,\theta \in \Gamma, \zeta \in \Sigma}: \widetilde{H} \in \widetilde{\mathcal{A}}^{\Theta},  \text{ and } \left( \widetilde{W}^{k,\theta, \zeta}_{\infty}(\widetilde{H}) \right)_{k \in D,\theta \in \Gamma, \zeta \in \Sigma} \text{exists}  \right\rbrace,\\ 
	\widetilde{\bC}^{D \times \Gamma \times \Sigma}_0 &=& (\widetilde{\bK}^{D \times \Gamma \times \Sigma}_0 - \bL^{0,D \times \Gamma \times \Sigma}_{\infty,+}), \label{eq:tilde_C_0}\\
	\widetilde{\bC}^{D \times \Gamma \times \Sigma} &=&  \widetilde{\bC}^{D \times \Gamma \times \Sigma}_0 \bigcap \mathbf{L}^{\infty, \mathbb{N} \times \Theta \times \Theta,b}_{\infty}. \label{eq:tilde_C}
\end{eqnarray}
\begin{remark}
\begin{itemize}
\item In \eqref{eq:K_x}, the random variable $\widetilde{W}^{k,\theta, \zeta}_{\infty}(\widetilde{H})$ is always well-defined if $T_k \wedge \zeta < \infty$. On the event $T_k \wedge \zeta = \infty,$ we define $$\widetilde{W}^{k,\theta, \zeta}_{\infty}(\widetilde{H}):= \lim_{t \to \infty} \int_0^{t}  \widetilde{H}_u(\theta)dS_u.$$
\item In the definition of $\widetilde{\bC}^{D \times \Gamma \times \Sigma}_0$, it is important to ask for the admissibility constraint $\widetilde{H} \in \widetilde{\A}^{\Theta}$ instead of the admissibility for all $\theta \in \Gamma$. This requirement is used later in Section \ref{sec:proof_thm:local_closed_0} (Proof of Theorem \ref{thm:local_closed_0}), where we need $\mathcal{H}^{k,\theta} \subset \overline{\mathcal{H}^{\mathbb{N} \times \Theta}}$ for any $k \in \mathbb{N}, \theta \in \Theta$. 
\end{itemize} 
\end{remark}
The local NFLVR$_{\mathfrak{C}}$  condition is reformulated as below.
\begin{proposition}
	Assumption \ref{assumption_main} (i) holds  if and only if  \begin{equation}\label{eq:nflvr1}
		\overline{\widetilde{\bC}^{D \times \Gamma \times \Sigma}}^{\|\|_{\infty}} \cap \bL^{\infty,D \times \Gamma \times \Sigma,b}_{\infty,+} = \{0\}, \  
	\end{equation}
	for all $\emptyset \ne D \times \Gamma \times \Sigma \subset \mathbb{N} \times \Theta \times \Theta,$ where $\overline{\widetilde{\bC}^{D \times \Gamma \times \Sigma}}^{\|\|_{\infty}}$ is the closure of $\widetilde{\bC}^{D \times \Gamma \times \Sigma}$ in the product space $\bL^{\infty,D \times \Gamma \times \Sigma,b}_{\infty}$. 
\end{proposition}
\begin{proof}
	$``\Rightarrow":$ Assume that \eqref{eq:nflvr1} fails. There exist $D \subset \mathbb{N}, \Gamma \subset \Theta, \Sigma \subset \Theta$, a vector $\bL^{\infty,D \times \Gamma \times \Sigma,b}_{\infty, +} \ni \mathbf{f} \ne \mathbf{0}$ and a sequence $\mathbf{f}_n \in \widetilde{\bC}^{D \times \Gamma \times \Sigma}, n \in \mathbb{N}$ such that $\mathbf{f}_n \to \mathbf{f}$ in $\bL^{\infty,D\times \Gamma, \Sigma,b}$ in the sup norm topology. Let $k \in D, \theta \in \Gamma, \zeta \in \Sigma$ be such that $P(f^{k,\theta, \zeta} >0) >0$. The corresponding sequence $f^{k,\theta, \zeta}_n, n \in \mathbb{N}$ satisfies
	$$\|f^{k,\theta, \zeta}_n - f^{k,\theta,\zeta} \|_{\infty} \to 0.$$ Up to a subsequence, we may assume $\|f^{k,\theta, \zeta}_n - f^{k,\theta, \zeta} \|_{\infty} \le  \frac{1}{n}$  and therefore, 
	$f^{k,\theta, \zeta}_n \ge f^{k,\theta, \zeta} - \frac{1}{n} \ge - \frac{1}{n}$. This means the sequence $f^{k,\theta, \zeta}_n, n \in \mathbb{N}$ is a FLVR for the market up to $T_k \wedge \zeta$, which contradicts Assumption \ref{assumption_main} (i). 
	
	$``\Leftarrow":$ Choosing $D = \{k\}, \Gamma = \{\infty\}, \Sigma = \{\infty\}$ for an arbitrary $k \in \mathbb{N}$ yields 
	\begin{equation}\label{na1}
		\overline{\widetilde{\bC}^{ \{k\} \times \{\infty\} \times \{\infty\} }} \cap \bL^{\infty,\{k\} \times \{\infty\}\times \{\infty\},b}_{\infty,+} = \{0\}.
	\end{equation}
	Assume that the condition NFLVR$_{\mathfrak{C}}$ for the market up to time $T_k$ fails, that is 
	$\overline{C^k}^{\|\|_{\infty}}  \cap L^{\infty}_+ \ne \{0\}$. There exists a sequence $H_n1_{[[0,T_k]]} \in \mathcal{A}_{1/n}, n \in \mathbb{N}$, such that $\int_0^{T_k}H_{n,u} dS_u$ converges to $0 \ne W \in [0, \infty)$. On the event $\{T_k = \infty\}$, the terminal wealth is well-defined by the definition of $C^k$. We construct a sequence of generalized strategies $\widetilde{H}_n = (H_n1_{[[0, T_k]]},(H_n(\theta))_{\theta \in \Theta })$ where $H_n(\theta) = H_n1_{[[0, T_k]]}, \ \theta \in \Theta$. By construction, $\widetilde{H}_n \in 	\widetilde{\mathcal{A}}^{\Theta}_{1/n}$ for each $n \in \mathbb{N}$. The corresponding sequence of terminal values  $\widetilde{W}^{k,\infty,\infty}_{\infty}(\widetilde{H}_n)$ converges to $W$, which contradicts \eqref{na1}.  \qed 
\end{proof}
The first main result is the $\bw^*$-closedness of $\widetilde{\bC}^{\mathbb{N} \times \Theta \times \Theta}$ whose proof is given in Subsection \ref{sec:proof_thm:local_closed_0}.  
\begin{theorem}\label{thm:local_closed_0}
	Under Assumption \ref{assumption_main} (i), for any $\emptyset \ne D \times \Gamma \times \Sigma \subset \mathbb{N} \times \Theta \times \Theta$ such that $\infty \in \Gamma$ and $\infty \in \Sigma$,  the convex cone $\widetilde{\bC}^{D \times \Gamma \times \Sigma}$ defined in \eqref{eq:tilde_C} is $\bw^*$-closed. 
\end{theorem} 
Using the $\bw^*$-closedness, we apply the typical separation argument and obtain the new pricing systems below.
\begin{definition}\label{def:pricing_func_American}
	Let $\mathcal{T}$ be a localizing sequence and $\Theta$ be defined as in (\ref{defi:theta}). Let $D = \{k_1 < ...< k_p\}$,  $\Gamma = \{\theta_1 < ... < \theta_q\}$, $\Sigma = \{ \zeta_1, < ... < \zeta_r\}$. A $(D \times \Gamma \times \Sigma)-$pricing system is a vector $\bZ \in \bigoplus_{(k,\theta, \zeta) \in D \times \Gamma \times \Sigma} L^1(\mathcal{F}_{T_k \wedge \zeta},P)$ such that
	\begin{itemize}
		\item[(i)] $0 \le \mathbf{Z} $ and  $E\left[ \sum_{(k,\theta, \zeta) \in D \times \Gamma \times \Sigma} Z^{k,\theta, \zeta} \right] =1;$
		\item[(ii)] For any $u \in \mathfrak{C}$, any two stopping times  $\theta_{\ell} < \sigma \le \tau \le \theta_{\ell+1}$ for $\ell \in \{0,...,q-1\}$ ($\theta_0 := 0$ by convention) and $B \in \mathcal{F}_{\sigma}$, we have \begin{eqnarray}\label{eq:mart_11}
			E\left[ \sum_{i=1}^{p} \sum_{s = 1}^{r} \sum_{j=\ell+1}^{q}  Z^{k_i,\theta_j, \zeta_s} \left( u \cdot S^{T_{k_i} \wedge \zeta_s}_{\tau} - u \cdot S^{T_{k_i} \wedge \zeta_s}_{\sigma} \right)1_{B}  \right] \le 0;	
		\end{eqnarray}
		\item[(iii)] For any $u \in \mathfrak{C}$, any $\theta_j \in \Gamma$, any two stopping times $\theta_j < \sigma \le \tau $ and $B \in \mathcal{F}_{\sigma}$, we have
		\begin{equation}\label{eq:martingale_after1}
			E\left[ \sum_{i=1}^{p} \sum_{s=1}^{r} Z^{k_i,\theta_j, \zeta_s} \left( u \cdot  S^{T_{k_i} \wedge \zeta_s}_{\tau} - u \cdot  S^{T_{k_i} \wedge \zeta_s}_{\sigma} \right)1_{B}  \right] \le 0.
		\end{equation}
	\end{itemize}
\end{definition}
For $k \in D, \theta \in \Gamma, \zeta \in \Sigma$, we define also the martingales $$Z^{k,\theta,\zeta}_t = E[Z^{k, \theta, \zeta}|\F_t].$$ 
We may compare the new pricing systems to absolutely continuous martingale measures introduced in \cite{ds_absolute}. For each $\theta_j \in \Gamma, $ the condition \eqref{eq:martingale_after1} describes the supermartingale property of the deflated wealth processes after the exercise time $\theta_j$. The condition \eqref{eq:mart_11} explains the supermartingale property of the deflated wealth processes between two exercise times $\theta_{\ell}, \theta_{\ell+1}$. Remark \ref{ex:density} below illustrates this effect.
Note that we could always embed a $(D \times \Gamma \times \Sigma)-$pricing system into a $(\mathbb{N}\times \Theta \times \Theta)-$pricing system by setting $Z^{k,\theta,\zeta} = 0$ for $(k,\theta, \zeta) \notin D \times \Gamma \times \Sigma$.  Conversely, a $(\mathbb{N}\times \Theta \times \Theta)-$pricing system is also a $(D \times \Gamma \times \Sigma)-$pricing system for some $D, \Gamma, \Sigma$ by the definition of direct sum.  The set of such pricing systems is denoted by $
\mathcal{Z}^{A}$, which is useful for computing superhedging prices of American options.

In particular, when $\Theta = \{\infty\}$, the concept of pricing system in Definition \ref{def:pricing_func_American} is reduced as follows. 
\begin{definition}\label{def:pricing_func}
Let $\mathcal{T}$ be a localizing sequence. A $(D\times\Sigma)-$ pricing system is a vector $\bZ \in \bigoplus_{k \in D, \zeta \in \Sigma} L^1(\mathcal{F}_{T_k \wedge \zeta},P)$ such that
	\begin{itemize}
		\item[(i)] $0 \le \mathbf{Z} $ and  $E\left[ \sum_{k \in D, \zeta \in \Sigma} Z^{k, \zeta} \right] =1$; 
		\item[(ii)] For any $u \in \mathfrak{C}$ and any two stopping times $0 < \sigma \le \tau$ and $B \in \mathcal{F}_{\sigma}$, we have 	$$E\left[ \sum_{k \in D, \zeta \in \Sigma} Z^{k,\zeta} \left( u \cdot 
		S^{T_k \wedge \zeta}_{\tau} - u \cdot  S^{T_k \wedge \zeta}_{\sigma} \right)1_{B}  \right] \le 0.$$	
	\end{itemize}
\end{definition}
The set of such pricing systems is denoted $\mathcal{Z}$, which is used for pricing European options. 
\begin{example}\label{example}
	Let $\mathfrak{C} = \mathfrak{C}^u$ and $d=1$.
	Assume $Q^{k}$ is an equivalent local martingale measure for the market up to $T_k$. For each $k \in \mathbb{N}$, denote $Z^{k}:=\left. dQ^{k}/dP\right|_{ T_k}$. The vector  $\mathbf{Z}=(0,...0,Z^{k},0,...)$ is a $(\{k\} \times \{\infty\})$-pricing system. The vector $\mathbf{Z}=0.5(0,...0,Z^{k_1},0...0,Z^{k_2},0...)$ is a $(\{k_1,k_2\} \times \{\infty\})$-pricing system.  We note that from such $(\{k\} \times \{\infty\})$-pricing systems $\bZ$, we can concatenate $Z^{k}, k \in \mathbb{N}$ to get an ELMD by standard arguments, see for example in  \cite{cdm15}. An ELMD is not an element in $\bigoplus_{k \in \mathbb{N}, \zeta \in \Theta} L^{1}$, and thus, it is not suitable for our computations in the product spaces.
\end{example}
\begin{remark}\label{ex:density}%Also, perhaps a remark should be added for the case $\mathfrak{C}=\mathbb{R}^d$ (in this case there is no need for $u$ in these equations). I would also add an explanation about all the three conditions (17), (18), (19) in a remark: they are difficult to understand.
	Consider the case without constraint $\mathfrak{C} = \mathfrak{C}^u, d= 2$. Let $0 < \sigma \le \tau$ be two stopping times and $B \in \mathcal{F}_{\sigma}$. Recall that  $\mathcal{S}^{1,T_k}_t = \mathcal{S}^{1}_{t \wedge T_{k}}, \mathcal{S}^{2,T_k}_t = \mathcal{S}^{2}_{t \wedge T_{k}} $. 

The case $D = \{k\}, \Gamma = \{\theta\}, \Sigma = \{\infty\}$:   Choosing $u \in \{ (\pm 1,0), (0, \pm 1)\}$, the condition \eqref{eq:mart_11} yields for $0 < \sigma \le \tau \le \theta$,
		\begin{eqnarray*}
			E\left[ \left(Z^{k,\theta,\infty}_{\tau} 
			\mathcal{S}^{1,T_k}_{\tau} - Z^{k,\theta,\infty}_{\sigma} \mathcal{S}^{1,T_k}_{\sigma} \right)1_{B}  \right] &=& 0,\\	 
			E\left[ \left(Z^{k,\theta,\infty}_{\tau} 
			\mathcal{S}^{2,T_k}_{\tau} - Z^{k,\theta,\infty}_{\sigma} \mathcal{S}^{2,T_k}_{\sigma} \right)1_{B}  \right] &=& 0.
		\end{eqnarray*} 
		Similarly, for $\theta < \sigma \le \tau < \infty$,  the condition \eqref{eq:martingale_after1} gives
		\begin{eqnarray*}
			E\left[ \left( Z^{k,\theta,\infty}_{\tau}\mathcal{S}^{1,T_k}_{\tau} - Z^{k,\theta,\infty}_{\sigma}\mathcal{S}^{1,T_k}_{\sigma} \right)1_{B}  \right] &=& 0,\\ E\left[ \left( Z^{k,\theta,\infty}_{\tau}\mathcal{S}^{2,T_k}_{\tau} - Z^{k,\theta,\infty}_{\sigma}\mathcal{S}^{2,T_k}_{\sigma} \right)1_{B}  \right] &=& 0.
		\end{eqnarray*}
		In other words, the products $Z^{k,\theta,\infty}_{t}\mathcal{S}^{1,T_k}_{t}, Z^{k,\theta,\infty}_{t}\mathcal{S}^{2,T_k}_{t}$ are martingales on $[0,\infty)$.

The case $D = \{k\}, \Gamma = \{\theta_1, \theta_2\}, \Sigma = \{\infty\}$: Again, the condition \eqref{eq:mart_11} yields for $0 < \sigma \le \tau \le \theta_1,$
		\begin{eqnarray*}
			E\left[ \left(\left( Z^{k,\theta_1,\infty}_{\tau} + Z^{k,\theta_2,\infty}_{\tau} \right)  
			\mathcal{S}^{1,T_k}_{\tau} - \left( Z^{k,\theta_1,\infty}_{\sigma} + Z^{k,\theta_2,\infty}_{\sigma} \right)  \mathcal{S}^{1,T_k}_{\sigma} \right)1_{B}  \right] &=& 0,	\\
			E\left[ \left(\left( Z^{k,\theta_1,\infty}_{\tau} + Z^{k,\theta_2,\infty}_{\tau} \right)
			\mathcal{S}^{2,T_k}_{\tau} - \left( Z^{k,\theta_1,\infty}_{\sigma} + Z^{k,\theta_2,\infty}_{\sigma} \right)  \mathcal{S}^{2,T_k}_{\sigma} \right)1_{B}  \right] &=& 0.
		\end{eqnarray*} 
		For $\theta_1 < \sigma \le \tau \le \theta_2$, we obtain
		\begin{eqnarray}
			E\left[ \left(  Z^{k,\theta_2,\infty}_{\tau}  
			\mathcal{S}^{1,T_k}_{\tau} -  Z^{k,\theta_2,\infty}_{\sigma} \mathcal{S}^{1,T_k}_{\sigma} \right)1_{B}  \right] &=& 0, \label{ex:density1}	\\
			E\left[ \left( Z^{k,\theta_2,\infty}_{\tau}
			\mathcal{S}^{2,T_k}_{\tau} -  Z^{k,\theta_2,\infty}_{\sigma}   \mathcal{S}^{2,T_k}_{\sigma} \right)1_{B}  \right] &=& 0. \label{ex:density31}
		\end{eqnarray} 
		The condition \eqref{eq:martingale_after1} leads to
		\begin{eqnarray}
			E\left[ \left( Z^{k,\theta_1,\infty}_{\tau}  
			\mathcal{S}^{1,T_k}_{\tau} -  Z^{k,\theta_1,\infty}_{\sigma}  \mathcal{S}^{1,T_k}_{\sigma}  \right)1_{B}  \right] &=& 0, \label{ex:density2} \\
			E\left[ \left( Z^{k,\theta_1,\infty}_{\tau}  
			\mathcal{S}^{2,T_k}_{\tau} -  Z^{k,\theta_1,\infty}_{\sigma}  \mathcal{S}^{2,T_k}_{\sigma}  \right)1_{B}  \right] &=& 0.  \label{ex:density4}
		\end{eqnarray}
		for $\theta_1 < \sigma < \tau < \infty$, and
		\begin{eqnarray*}
			E\left[ \left( Z^{k,\theta_2,\infty}_{\tau}  
			\mathcal{S}^{1,T_k}_{\tau} -  Z^{k,\theta_2,\infty}_{\sigma}  \mathcal{S}^{1,T_k}_{\sigma} \right)1_{B}  \right] &=& 0, \\
			E\left[ \left( Z^{k,\theta_2,\infty}_{\tau}  
			\mathcal{S}^{2,T_k}_{\tau} -  Z^{k,\theta_2,\infty}_{\sigma}  \mathcal{S}^{2,T_k}_{\sigma} \right)1_{B}  \right] &=& 0,
		\end{eqnarray*}
		for $\theta_2 < \sigma < \tau < \infty$. Note that \eqref{ex:density1}  differs from \eqref{ex:density2} and \eqref{ex:density31} differs from  \eqref{ex:density4}.  

The case $D = \{k_1,...,k_p\}, \Gamma = \{\theta_1,..., \theta_q\}, \Sigma = \{\infty\}$: Again, for $\theta_{\ell} < \sigma \le \tau \le \theta_{\ell + 1}, \ell \in \{0,...,q-1\}$, the condition \eqref{eq:mart_11} yields 
		\begin{eqnarray*}
			E\left[ \sum_{i=1}^{p}\left(\sum_{j=\ell+1}^{q}Z^{k_i,\theta_j,\infty}_{\tau} \mathcal{S}^{1,T_{k_i}}_{\tau} 
			- \sum_{j=\ell+1}^{q}Z^{k_i,\theta_j,\infty}_{\sigma} \mathcal{S}^{1,T_{k_i}}_{\sigma} \right)1_{B}  \right] &=& 0,	\\
			E\left[ \sum_{i=1}^{p}\left(\sum_{j=\ell+1}^{q}Z^{k_i,\theta_j,\infty}_{\tau} \mathcal{S}^{2,T_{k_i}}_{\tau} 
			- \sum_{j=\ell+1}^{q}Z^{k_i,\theta_j,\infty}_{\sigma} \mathcal{S}^{2,T_{k_i}}_{\sigma} \right)1_{B}  \right] &=& 0.
		\end{eqnarray*} 
		For each $\theta \in \Gamma,$ the condition \eqref{eq:martingale_after1} leads to
		\begin{eqnarray*}
			E\left[  \sum_{i=1}^{p} \left(  Z^{k_i,\theta,\infty}_{\tau}
			\mathcal{S}^{1,T_{k_i}}_{\tau} -  Z^{k_i,\theta,\infty}_{\sigma}  \mathcal{S}^{1,T_{k_i}}_{\sigma}  \right)1_{B}  \right] &=& 0,\\
			E\left[  \sum_{i=1}^{p} \left(  Z^{k_i,\theta,\infty}_{\tau}
			\mathcal{S}^{2,T_{k_i}}_{\tau} -  Z^{k_i,\theta,\infty}_{\sigma}  \mathcal{S}^{2,T_{k_i}}_{\sigma}  \right)1_{B}  \right] &=& 0, \label{ex:density3}
		\end{eqnarray*}
		for $\theta < \sigma < \tau < \infty$.
\end{remark}
\begin{remark}\label{ex:density_short} Consider the setting in Remark \ref{ex:density} for the case with no short sale constraint for the stock $\mathcal{S}^{1}$. In this situation, all the equalities in the equations for $\mathcal{S}^{1}$ are replaced by the inequality $\le$, and we get the  supermartingale properties for the corresponding deflated versions of $\mathcal{S}^{1}$.
\end{remark}
We state a version of the first Fundamental Theorem of Asset Pricing. This time it is formulated in terms of pricing systems instead of martingale measures, noting that generalized strategies are used for trading.
\begin{corollary}\label{cor:density_A}[FTAP] Let $\mathfrak{C} = \mathfrak{C}^u$ or $\mathfrak{C} = \mathfrak{C}^s$. The following are equivalent:
	\begin{itemize}
		\item[(i)] Assumption \ref{assumption_main} (i) holds;
		\item[(ii)] for any $(k,\theta, \zeta) \in D \times \Gamma \times \Sigma$ such that $\infty \in \Gamma, \infty \in \Sigma$ and $A \in \mathcal{F}_{T_{k}\wedge \zeta}$ with $P(A) > 0$, there exists a $(D \times \Gamma \times \Sigma)-$pricing system satisfying $E[Z^{k,\theta,\zeta}1_{A}] > 0.$
		\item[(iii)] for any $(k, \zeta) \in D  \times \Sigma$ such that $\infty \in \Sigma$ and $A \in \mathcal{F}_{T_{k}\wedge \zeta}$ with $P(A) > 0$, there exists a $(D\times \Sigma)-$pricing system  satisfying $E[Z^{k,\zeta}1_{A}] > 0.$
	\end{itemize}
\end{corollary}
\begin{remark}
We only use the local boundedness of $S$ in the proof of Corollary \ref{cor:density_A} when constructing admissible strategies. Theorem \ref{thm:local_closed_0} remains valid even the local boundedness property is dropped. This is similar to the results in  \cite{kabanov}. We believe that extending other results to the case with    unbounded $S$ is possible and one has to take care of big jumps of $S$ as in the proof of Main Theorem 14.1.1 of \cite{delbaen2006mathematics}. We leave this extension for future studies, because the key ideas in the present paper will certainly be lost under the technical complexity involving sigma-martingales for the unbounded case.
\end{remark}
In discrete time settings, versions of FTAPs under cone constraints were given in \cite{pt99}, \cite{napp}. In Application 3.2 of \cite{pt99}, the condition no time $t$ local arbitrage is equivalent to the existence of a $\mathcal{F}_t$-measurable density $Z$ such that 
\begin{equation}\label{eq:pham_touzi}
	E[Z(diag(S_{t-1})^{-1}S_t - 1)|\mathcal{F}_{t-1}] \in \widetilde{\mathfrak{C}}, 
\end{equation}
where $\widetilde{\mathfrak{C}}: = \{x \in \mathbb{R}^d: x \cdot y \le 0, \forall y \in \mathfrak{C}\}$. The conditions \eqref{eq:mart_11}, \eqref{eq:martingale_after1} are close in spirit to \eqref{eq:pham_touzi}, however, \cite{pt99} considered the constraint $diag(S_t) H_t \in \mathfrak{C}$ instead.

The new pricing system also gives supermartingale property for wealth processes in the product framework.
\begin{proposition}\label{pro:1}
	Let $\mathfrak{C} = \mathfrak{C}^u$ or $\mathfrak{C} = \mathfrak{C}^s$. Let $\bZ$ be a $(D\times \Gamma \times \Sigma)-$pricing system. For any $\widetilde{H} \in \widetilde{\mathcal{A}}^{\Theta}$, we have 
	$$E\left[\sum_{(k,\theta, \zeta) \in D \times \Gamma \times \Sigma}Z^{k,\theta,\zeta} \widetilde{W}^{k,\theta,\zeta}_{\infty}(\widetilde{H}) \right] \le 0.$$
\end{proposition}
\begin{remark}
	The supermartingale property can be obtained from nonnegative wealth processes and ELMDs. For some $H \in \mathcal{A}_x$ and an ELMD $Y$, the local martingale $Y_t\left( x + H \cdot S_t \right)$ is bounded from below by zero and hence a supermartingale. We get that 	\begin{equation}\label{eq:supermart} 		E\left[ Y_T \left( x + H \cdot S_{T} \right)  \right] \le x, \ \forall T > 0,
	\end{equation} 	
	and thus $$E\left[Y_T \left(H \cdot S)_{T} \right)  \right] \le x - xE[Y_T].$$
	When $Y$ is a strict local martingale, the upper bound may be different from zero. In that case, $Y_t\left(H \cdot S_t \right)$ is not a supermartingale, as expected under arbitrage. 
\end{remark}
\begin{remark}
	%Last remark: I think that at some point there should be a discussion about the note, that is, in previous work they were working with $C_x^x$ while we can work with $C_x^z$. (Clearly, this should be explained without introducing new letters.)
	For non-negative wealth processes with an initial capital $z > 0$, the corresponding strategies must be in $\mathcal{A}_z$. The closedness of the set of non-negative wealth processes with the initial capital $z$ in the semimartingale topology is proved in  \cite{Kardaras} by a change of num\'eraire technique. The author discussed in Remark 1.12 that the supermartingale property of discounted processes is not suitable when wealth is negative and asked for different techniques to work with negative wealth processes. In this paper, we investigate the situation where the initial capital $z$ differs from the credit constraint $x$, i.e., the wealth processes can be negative as long as they are uniformly bounded below by $-x$. Therefore, the set of wealth processes in our setting is larger than that of \cite{Kardaras}, which may lead to higher expected utilities and smaller hedging prices.  The supermartingale property is preserved by the new pricing systems, as in Proposition \ref{pro:1}.
\end{remark}

Our second main result of the paper is the superhedging duality for American options. 
\begin{theorem}\label{thm:hedge_American}
	Let $\mathfrak{C} = \mathfrak{C}^u$ or $\mathfrak{C} = \mathfrak{C}^s$. Let Assumption \ref{assumption_main} be in force. Assume that $\Phi$ is bounded or $\Phi \ge 0$. Let $x>0$ be a fixed credit constraint such that $\Phi_{t} \ge -x, \ a.s.$, for any $t \ge 0$. There exists $(z,\widetilde{H})\in \mathbb{R} \times \widetilde{\mathcal{A}}^{\mathbb{T}}_{}$ such that
	$$
z + \int_{0}^{t}\widetilde{H}_u(\tau)\cdot dS_u - \Phi_{\tau}1_{\tau \le t} \ge -x1_{t < \infty}, \ a.s. \ \forall \tau \in \mathbb{T}, t \in [0,\infty]
	$$
	if and only if
	\begin{equation}\label{eq:duality}
		z\geq\sup_{\bZ\in \mathcal{Z}^{A}}E\left[\sum_{k \in \mathbb{N}, \theta\in \Theta, \zeta \in \Theta}Z^{k,\theta,\zeta} \left( \Phi_{\theta}1_{\theta\le  \zeta \wedge T_k } - x 1_{\zeta \wedge T_k <\infty} \right) \right].
	\end{equation} 
\end{theorem}
The proof of this theorem is given in Section \ref{subsection:proof_American}. 
Other ideas to transform problems with American options into ones with European options were discussed in literature, for example, in \cite{vek}, \cite{adot}. It is also known that we need to enlarge the probability space to include uncertain exercise times. Under transaction costs, \cite{vek} introduced  $\mathfrak{T}^m:=\{t_k^m =k2^{-m}T,k=0,...,2^m\}$ and equipped the product space $\Omega \times [0, T ]$ with the product measure $P \otimes \nu^m$ where $T$ is the time horizon and the measure $\nu^m$ charges only the points of $\mathfrak{T}^m$ with equal weights $1/(2^m + 1)$. The concept of  Fatou-convergence in $L^0(P \otimes \nu^m)$ therein and our Fatou-convergence are very much similar. As explained in \cite{vek}, ``the expected ``value" of an American claim is an expectation of the weighted average of ``values" of assets obtained by the option holder for a variety of exercise dates". The authors in \cite{vek} therefore introduced a specific class of coherent price
systems which perfoms this idea. In a discrete time setting $\mathfrak{T} = \{1,2.,...N\}$, the paper \cite{adot} employed the enlarged probability space $\Omega \times \mathfrak{T}$ and showed that superhedging prices of American options are exactly superhedging prices of the corresponding European options in the enlarged space, and the superhedging duality holds true. Our approach shares similar spirits to the two mentioned work, or in terms of mathematics, this could be seen by comparing the enlarged space $L^0(P \otimes \nu^m)$ to the product space $\prod_{\theta \in \mathfrak{T}^m}L^0(P)$. However, there are some features that distinguish our approach from the mentioned ones. First, the approach with product spaces naturally incorporates stopping times under the local NFLVR$_{\mathfrak{C}}$ condition, while the others may need more effort to work with stopping times. Secondly, it is possible to equip suitable topologies in continuous time settings (as the $\bw^*$ topology in our framework), and to give financial and geometrical meanings, in analogous to the classical settings where NFLVR$_{\mathfrak{C}}$ holds globally.

A superhedging duality was also given in the book \cite{kk2021} where the case with continuous asset prices and nonnegative hedging portfolios is treated. In comparison, the superhedging price in the present paper may be smaller than that of \cite{kk2021}, because the sellers could simultaneously exploit better credit constraints and arbitrage profits after the exercise time. On the other hand, the authors in  \cite{kk2021} followed a stochastic analysis approach and employed the optional decomposition theorem for
continuous semimartingales in general filtrations under the existence of strictly positive local martingale deflators, see also \cite{kk15}. A general optional decomposition was given in  \cite{sy1998}. Our approach is completely different and relies on functional analysis, instead. 

\section{Pricing and hedging European options}\label{sec:eur}
Superhedging dualities for European options using 
nonnegative portfolios were discussed in Section 4.7 of \cite{kk}, and in \cite{khasanov}  by the change of num\'eraire technique (without the existence of optimal hedging strategies), or in \cite{kk2021} by a decomposition theorem. In this section, we explain how to use the present approach for European options. Let $G_{T^{\sharp}}$ be a European option with maturity $T^{\sharp} < \infty$. By Proposition \ref{pro:reduce}, in order to superhedge $G_{T^{\sharp}}$, we need to find a capital $z$ and a strategy $H \in \A$ (here, there is no need to use a generalized strategy) such that $\forall t \in [0,\infty]$
\begin{equation}\label{eq:super_Eur}
z + \int_{0}^{t}H_u dS_u - G_{T^{\sharp}}1_{T^{\sharp} \le t} \ge -x1_{t < \infty}, \ a.s., 
\end{equation}
or equivalently, to find a capital $z$ and a strategy $H \in \A$ satisfying for all $k \in \mathbb{N}, \zeta \in \Theta$,
\begin{equation}\label{eq:super_Eur_prod}
z + \int_0^{\zeta \wedge T_k} H_u dS_u - G_{T^{\sharp}}1_{T^{\sharp}\le  \zeta \wedge T_k }  \ge  - x 1_{\zeta \wedge T_k <\infty}, \ a.s.
\end{equation}
Therefore, a superhedging duality for European options can be deduced from Theorem  \ref{thm:hedge_American}. 
\begin{theorem}\label{thm:hedge_Eur}
	Let $\mathfrak{C} = \mathfrak{C}^u$ or $\mathfrak{C} = \mathfrak{C}^s$. Let Assumption \ref{assumption_main} be in force. Let $G_{T^{\sharp}}$ be a $\mathcal{F}_{T^{\sharp}}$-measurable contingent claim such that $G_{T^{\sharp}}$ is bounded or $G_{T^{\sharp}}\ge 0$. Let $x>0$ be a fixed credit constraint such that $G_{T^{\sharp}} \ge -x$. There exists $(z,H)\in \mathbb{R} \times \mathcal{A}_{}$ such that $\forall t \in [0,\infty]$
	$$
	z + \int_{0}^{t}H_u dS_u - G_{T^{\sharp}}1_{T^{\sharp} \le t} \ge -x1_{t < \infty}, \ a.s.
	$$
	if and only if
	\begin{equation}\label{eq:duality_Eur}
		z\geq\sup_{\bZ\in \mathcal{Z}}E\left[\sum_{k \in \mathbb{N},\zeta \in \Theta}Z^{k,\zeta}(G_{T^{\sharp}}1_{T^{\sharp}\le  \zeta \wedge T_k }  - x 1_{\zeta \wedge T_k <\infty}) \right].
	\end{equation}
	We denote by $\pi_{\infty}(G_{T^{\sharp}},x)$ the minimal superhedging price of $G_{T^{\sharp}}$ with the credit constraint $x$. 
\end{theorem}
It is observed from \eqref{eq:duality_Eur} that the supremum is obtained when the pricing systems put more weight on $Z^{k, \zeta}$ such that $T_k \wedge \zeta \ge T^{\sharp}$. If there is a trading horizon $T<\infty$ and $T^{\sharp} \le T$, we will embed the finite horizon market $[0,T]$ into the infinite horizon setting by stopping the stock processes at time $T$. We write $\pi_{\infty}(G_{T^{\sharp}},x) = \pi_{T}(G_{T^{\sharp}},x)$.  

If $T^{\sharp} = T < \infty$, there are no trading activities after the option maturity. Choosing $\Sigma = \{T\}$, the hedging requirement \eqref{eq:super_Eur_prod} becomes
%\begin{equation}\label{eq:super_Eur_2}z + \int_{0}^{t}H_u dS_u - G_{T}1_{T \le t} \ge -x1_{t < T}, \ a.s.,  t \in [0,T], \end{equation}
\begin{equation}\label{eq:super_Eur_prod_2}
	z + \int_0^{T \wedge T_k} H_u dS_u - G_{T}1_{T \le T_k }  \ge  - x 1_{ T_k < T}, \ a.s.
\end{equation} for all $k \in \mathbb{N}$.  The superhedging duality in Theorem \ref{thm:hedge_Eur} can be stated by using ELMDs only as follows. 
\begin{proposition}\label{pro:appro_hedging}
Let $\mathfrak{C} = \mathfrak{C}^u$ and $T < \infty$ be the trading horizon. Let $G_{T^{\sharp}}$ be a European option with maturity $T^{\sharp} =T$ satisfying the conditions in  Theorem \ref{thm:hedge_Eur}. Let $\pi_{T}(G_{T},x)$ be the superhedging  price of the option $G_{T}$ with the credit constraint $x$. Let $\pi_{k}(G_{T},x)$ be the superhedging price of the payoff $G_{T}1_{T \le T_k} - x1_{T_k < T}$ when trading until $T \wedge T_k$. Then 
	\begin{equation}\label{eq:superheding}
		\pi_{T}(G_{T},x) =  \lim_{k \to \infty}\pi_{k}(G_{T},x) =  \sup_{Y \in ELMD} \left(  E\left[ Y_{T}G_{T}\right]  - x(1 - E\left[ Y_{T}\right] ) \right).	
	\end{equation}
\end{proposition}
In \eqref{eq:superheding}, if $x = 0$, superhedging portfolios are always nonnegative and the price $\pi_{T}(G,0)$ is well-known, see \cite{kk}, \cite{kk2021}. Under the presence of arbitrage opportunities, the price $\pi_{T}(G,x)$ could be smaller than the price $\pi_{T}(G,0)$ in general, because sellers are allowed to use the credit constraint $x$. If the global condition NFLVR holds on $[0,T]$, the formula \eqref{eq:superheding} becomes the classical superhedging duality, where $x$ plays no roles in the price. When $T^{\sharp} <T$, it is possible that $\pi_{T}(G_{T^{\sharp}},x) \le \pi_{T^{\sharp}}(G_{T^{\sharp}},x)$ because sellers simultaneously use the credit constraint and continue trading to exploit arbitrage after the option maturity. Additionally, we have to use the new pricing systems to compute $\pi_{T}(G_{T^{\sharp}},x)$ as in \eqref{eq:duality_Eur} for such cases.
\subsection{Future extensions}
There are many possible extensions 
for the present paper. The risky assets are assumed to be locally bounded. It 
would be interesting to extend the FTAP results to the case with unbounded risky assets. 
For this goal, techniques developed in \cite{ds1998} should be adopted. The results should be compared to that of  \cite{takaoka}.
%Although 
%continuity conditions for American options' payoffs are sufficient for financial applications, it 
%is natural to ask if such conditions could be weaken. 

The present framework could be applied to situations with uncertain exercise times such as game options, contracts with 
multi-exercise opportunities such as swing options in commodity markets. It could be applicable 
when stochastic factors and trading frictions are taken into account.  Furthermore, 
it would be interesting to develop similar framework for large markets such as bond markets, 
where possibly nonpositive financial cash flows need to be priced without assuming absence of arbitrage.

\section{Proofs}\label{sec:proofs}
\subsection{Some closedness results}
\begin{lemma}\label{lemma:c-bounded}
Let Assumption \ref{assumption_main}(i) be in force.  For any $x \ge 0,$ the set $\widetilde{\bK}^{\mathbb{N} \times \Theta \times \Theta}_{x}$ defined in \eqref{eq:K_x} is $c$-bounded. 
\end{lemma}
\begin{proof}
	The set $\mathbb{N} \times \Theta \times \Theta$ is countable. We need to prove that for each $(k,\theta, \zeta) \in \mathbb{N} \times \Theta \times \Theta$, the set $\left\lbrace \widetilde{W}^{k,\theta, \zeta}_{\infty}(\widetilde{H}), \widetilde{H} \in  \widetilde{\mathcal{A}}^{\Theta}_x \right\rbrace $ is bounded in $L^0_+.$ This comes from Corollary 9.3.4 of \cite{delbaen2006mathematics}, since the market satisfies the condition NFLVR up to $T_k \wedge \zeta$. \qed 
\end{proof}
We recall Corollary 4.9 of \cite{cs11}. 
\begin{lemma}\label{lem:semi_closed}
	Let $K \in \mathbb{R}^d$ be a closed convex cone. The set $$\left\lbrace (H \cdot Y_t)_{t \ge 0}:  H \in \mathcal{L}(Y), H_t(\omega) \in K \right\rbrace $$ is closed in the semimartingale topology for all $\mathbb{R}^d$-valued semimartingales $Y$ if and only if $K$ is polyhedral.
\end{lemma}

\begin{definition}
Let $D$ be a countable set. We say that a sequence $\bff_n, n \in \mathbb{N}$ in $\mathbf{L}^{0,D,b}$ Fatou-converges to $\bff$ if for each $k \in D$, $f^{k}_n$ converges to $f^{k}$ a.s. and $f^{k}_n \ge -x, \ a.s.$ for some $x >0$. 
\end{definition}

Let $\tau$ be a stopping time  taking values in $[0,\infty]$ and $D \times \Sigma \subset \mathbb{N} \times \Theta$. Let $\bv= (v^{k, \tau, \zeta})_{k \in D, \zeta \in \Sigma}$ be $\mathcal{F}_{\tau} $-measurable. Define 

\begin{eqnarray}
	\mathbf{C}^{D \times \Sigma} (\tau, \mathbf{v}) &:=& \left\lbrace \left(  v^{k,\tau, \zeta} + \int_{\tau }^{\infty} H_u 1_{]]\tau,\zeta \wedge T_k]]}dS_u  - h^{k,\tau,\zeta}\right)_{k \in D, \zeta \in \Sigma} \right.   \nonumber\\
	&& \qquad \left. : H1_{]]\tau,\zeta \wedge T_k]]}  \in  \mathcal{A}(\tau  ,v^{k,\tau,\zeta}), h^{k,\tau,\zeta} \in L^{0,}_+\right\rbrace.	 \label{eq:c_set}
\end{eqnarray}
Denote  $B^{\infty}_{x} =\{ y \in L^{\infty}(\mathcal{F}_T,P): \|y\|_{L^{\infty}} \le x  \}$. 
The following proposition generalizes Theorem 4.2 (i) of \cite{delbaen1994general}. 
\begin{proposition}\label{pro:countable_Fatou}
	Let Assumption \ref{assumption_main}(i) hold. For any nonempty (possibly countable) set $D \times \Sigma \subset \mathbb{N} \times \Theta $,   the set $\bC^{D \times \Sigma}(\tau, \mathbf{v})$ defined in \eqref{eq:c_set}  is Fatou-closed, that is $\bC^{D \times \Sigma}(\tau, \mathbf{v}) \cap \prod_{k \in D, \zeta \in \Sigma} B^{\infty}_{x}$ is closed in the space $\bL^{0,D \times \Sigma,b}$.
\end{proposition}
\begin{proof}
	Let $\bff_n, n \in \mathbb{N}$ be a sequence in $\bC^{D \times \Sigma}(\tau,\bv)$ such that for every $k \in D, \zeta \in \Sigma$, $$f^{k,\zeta}_n \to f^{k,\zeta} \text{ in probability and } f^{k,\zeta}_n \ge - x,\ a.s.$$ Without loss of generality, we may assume that $x = 1$. We need to find $H$ such that 
\begin{eqnarray*}
	v^{k,\tau,\zeta} + \int_{\tau }^{\infty} H_u 1_{]]\tau,\zeta \wedge T_k]]}dS_u &\ge& f^{k, \zeta}, \ a.s.\\
H1_{]]\tau,\zeta \wedge T_k]]}  &\in&  \mathcal{A}(\tau  ,v^{k,\tau,\zeta}),
\end{eqnarray*}
for every $k \in D, \zeta \in \Sigma.$ 
	
By taking a subsequence, we may assume that $f^{k, \zeta }_n \to f^{k, \zeta}, \ a.s.$ for all $k \in D, \zeta \in \Sigma$, see Lemma A.1 of \cite{chau2020}. By definition, there is a sequence $H_n, n \in \mathbb{N}$ such that for any $k \in D, \zeta \in \Sigma$,
\begin{eqnarray}
v^{k,\tau,\zeta} + \int_{\tau }^{ \infty} H_{n,u} 1_{]]\tau,\zeta \wedge T_k]]}dS_u \ge f^{k,\zeta}_n &\ge& -1, \ a.s., \label{eq:1}\\
H_n1_{]]\tau,\zeta \wedge T_k]]}  &\in&  \mathcal{A}(\tau  ,v^{k,\tau,\zeta}). \nonumber
\end{eqnarray}
Since the condition NFLVR holds up to $T_k, k \in D$, we get 
$$H_{n} 1_{]]\tau,\zeta \wedge T_k]]} \in  \mathcal{A}_1(\tau ,v^{k,\tau,\zeta}), \ \forall k \in D, \zeta \in \Sigma,$$
by Proposition 9.3.6 of \cite{delbaen2006mathematics}. Define
	\begin{eqnarray*}
&&\mathfrak{D}(\tau,\bv) := \left\{ (g^{k,\zeta})_{k \in D, \zeta \in \Sigma} :  \exists H_{n}, n \in \mathbb{N}  \text{ such that } H_n1_{]]\tau,\zeta \wedge T_k]]} \in  \mathcal{A}_1(\tau ,v^{k,\tau,\zeta}),    \right.\\
&&\left. v^{k,\tau,\zeta} + \int_{\tau }^{\infty} H_{n,u} 1_{]]\tau,\zeta \wedge T_k]]} dS_u \to  g^{k, \zeta}, \text{and } g^{k,\zeta} \ge f^{k,\zeta}, \ a.s., \forall k \in D, \zeta \in \Sigma  \right\} .
	\end{eqnarray*}
	By Lemma 4.4 of \cite{cfr}, there exist $H^{\sharp}_n \in conv\{H_n,H_{n+1},...\}$ such that for all $k \in D, \zeta \in \Sigma$, we have  
	$$v^{k, \tau, \zeta} + \int_{\tau }^{\infty} H^{\sharp}_{n,u} 1_{]]\tau,\zeta \wedge T_k]]} dS_u  \to q^{k,\zeta} \ge f^{k,\zeta}, \ a.s.$$ for some $q^{k,\zeta}$ taking finite values.  So the set $\mathfrak{D}(\tau,\bv)$ is non-empty.  It is also closed in $\mathbf{L}^{0,D \times \Sigma}$. By Lemma \ref{lemma_maximal}, the set  $\mathfrak{D}(\tau,\bv)$ has a maximal element $\mathbf{g}_0$. It remains to check that there exists $H^*$ such that 
\begin{eqnarray}
v^{k,\tau,\zeta} + \int_{\tau}^{\infty} H^*_u 1_{]]\tau,\zeta \wedge T_k]]} dS_u  &=& g^{k,\zeta}_0, \ a.s. \label{eq:goal0}\\
H^*1_{]]\tau,\zeta \wedge T_k]]}  &\in&  \mathcal{A}_1(\tau  ,v^{k,\tau,\zeta}),\nonumber
\end{eqnarray}
for all $k \in D, \zeta \in \Sigma$.  Since $\mathbf{g}_0 \in \mathfrak{D}(\tau,\bv)$, there exists a sequence $\widehat{H}_n$ such that 
\begin{eqnarray}
v^{k, \tau, \zeta} + \int_{\tau }^{\infty} \widehat{H}_{n,u}  1_{]]\tau,\zeta \wedge T_k]]} dS_u &\to&  g^{k,\zeta}_0 \ge f^{k,\zeta}, \ a.s., \label{eq:hat_H} \\
\widehat{H}_n1_{]]\tau,\zeta \wedge T_k]]}  &\in&  \mathcal{A}(\tau  ,v^{k,\tau,\zeta}),\nonumber
\end{eqnarray}
for all $k \in D, \zeta \in \Sigma$. Fix $k_0 \in D, \zeta_0 \in \Sigma$ arbitrarily. We prove that \begin{equation}\label{eq:prob_conv}
		\sup_{\tau \le t < \infty}\left| \int_{\tau }^{t} \widehat{H}_{n,u}  1_{]]\tau,\zeta_0 \wedge T_{k_0}]]}  dS_u - \int_{\tau }^{t} \widehat{H}_{m,u}  1_{]]\tau,\zeta_0 \wedge T_{k_0}]]}  dS_u\right| \to 0, 
	\end{equation}
in probability as $n, m$ tend to infinity. If the claim fails, there exist subsequences $n_{\ell}, m_{\ell}$ such that
	$$ P\left( 	\sup_{\tau \le t < \infty }\left| \int_{\tau }^{t} \widehat{H}_{n_{\ell},u}   1_{]]\tau,\zeta_0 \wedge T_{k_0}]]}  dS_u - \int_{\tau }^{t} \widehat{H}_{m_{\ell},u}  1_{]]\tau,\zeta_0 \wedge T_{k_0}]]}  dS_u\right|> \alpha  \right) \ge \alpha $$
	for some $\alpha > 0$. Define 
	$$ U_{\ell} = \inf \left\lbrace  t \ge \tau:   \int_{\tau }^{t} \widehat{H}_{n_{\ell},u}    dS_u - \int_{\tau }^{t} \widehat{H}_{m_{\ell},u}   dS_u\ \ge \alpha \right\rbrace .$$
	Then $P(U_{\ell}  \le \zeta_0 \wedge T_{k_0} ) \ge \alpha$,  for every $\ell \in \mathbb{N}$. Define 
	\begin{equation}
	H^{\flat}_{\ell,u} = \widehat{H}_{n_{\ell},u}1_{u \le U_{\ell}} +  \widehat{H}_{m_{\ell},u} 1_{u > U_{\ell}}.
	\end{equation} 
We can check that $H^{\flat}_{\ell} 1_{]]\tau,\zeta \wedge T_k]]}  \in  \mathcal{A}_1(\tau  ,v^{k,\tau,\zeta})$ for all $k \in D, \zeta \in \Sigma$, by Lemma 2.3 (b) of \cite{kabanov}. We obtain that for any $k \in D, \zeta \in \Sigma, \ell \in \mathbb{N}$,
	\begin{eqnarray*}
	&&\int_{\tau }^{\infty}H^{\flat}_{\ell,u}  1_{]]\tau,\zeta \wedge T_k]]}  dS_{u} \nonumber \\
	&=& 1_{U_{\ell} \ge \zeta \wedge T_k} \int_{\tau }^{\infty }\widehat{H}_{n_{\ell}}  1_{]]\tau,\zeta \wedge T_k]]}  dS_{u}  \\
		&& + 1_{U_{\ell} < \zeta \wedge T_k}\left( \int_{\tau }^{U_{\ell} }\widehat{H}_{n_{\ell}}  1_{]]\tau,\zeta \wedge T_k]]}  dS_{u} +  \int_{\tau }^{\infty }\widehat{H}_{m_{\ell}}  1_{]]\tau,\zeta \wedge T_k]]}  dS_{u} \right.  \\
		&& \left. - \int_{\tau }^{U_{\ell} }\widehat{H}_{m_{\ell}}  1_{]]\tau,\zeta \wedge T_k]]} dS_{u} \right) \\
		&=& 1_{U_{\ell} \ge \zeta \wedge T_k} \int_{\tau}^{\infty }\widehat{H}_{n_{\ell}}  1_{]]\tau,\zeta \wedge T_k ]]}    dS_{u}  + 1_{U_{\ell} < \zeta \wedge T_k } \int_{\tau }^{\infty}\widehat{H}_{m_{\ell}}  1_{]]\tau,\zeta \wedge T_k ]]}  dS_{u} + \xi^{k, \zeta}_{\ell},
	\end{eqnarray*}
	where 
	$$\xi^{k,\zeta}_{\ell} = 1_{U_{\ell} < \zeta \wedge T_k}\left( \int_{\tau }^{U_{\ell} }\widehat{H}_{n_{\ell}}  1_{]]\tau,\zeta \wedge T_k ]]} dS_{u}  - \int_{\tau}^{U_{\ell} }\widehat{H}_{m_{\ell}}  1_{]]\tau,\zeta \wedge T_k ]]} dS_{u}\right)  \ge 0.$$
For all $k \in D, \zeta \in \Sigma$, the final payoffs are improved if $U_{\ell} < \zeta \wedge T_k$. However, the improvement occurs at $k_0, \zeta_0$ is suitable for our argument. Since for every $\ell \in \mathbb{N}$,
	$$P(\xi^{k_0, \zeta_0}_{\ell} \ge \alpha)  =  P(U_{\ell} < \zeta_0 \wedge T_{k_0}) \ge \alpha,$$
	there exist convex combinations $\xi^{k,\zeta,\sharp}_{\ell} \in conv\{\xi^{k,\zeta}_{\ell}, \xi^{k,\zeta}_{\ell+1},...\}$ such that $\xi^{k,\zeta,\sharp}_{\ell} \to \eta^{k,\zeta} \ge 0$ for all $k \in D, \zeta \in \Sigma$ 
	with $\eta^{k_0,\zeta_0} \ne 0,$ 
	by Lemma 9.8.1 of \cite{delbaen2006mathematics} and Lemma 4.4 of \cite{cfr}. The corresponding convex combinations $H^{\flat,\sharp}_{\ell}$ satisfies  
	$$v^{k,\tau,\zeta} + \int_{\tau }^{\infty} H^{\flat,\sharp}_{\ell,u}  1_{]]\tau,\zeta \wedge T_k]]}  dS_{u} \to g^{k,\zeta}_0 + \eta^{k,\zeta}, \ \forall k \in D, \zeta \in \Sigma,$$
	as $\ell \to \infty$. Since $\eta^{k_0,\zeta_0} \ne 0,$ the random vector  $\mathbf{g}_0$ is not a maximal element of $\mathfrak{D}(\tau,\bv)$, which is a contradiction.  
	
	From \eqref{eq:prob_conv}, because of the arbitrariness of $k_0, \zeta_0$, we conclude that 
	$$\sup_{n \in \mathbb{N}}  \sup_{\tau  \le t <\infty }\left| \int_{\tau }^{t} \widehat{H}_{n,u}  1_{]]\tau,\zeta \wedge T_k ]]}  dS_u \right|  < \infty, \ a.s., $$
 for all $k \in D, \zeta \in \Sigma$. Define
	$$B^{D \times \Sigma} = \bigcup_{k \in D, \zeta \in \Sigma}]]\tau, \zeta \wedge T_k]]$$ and redefine  $\widehat{H}_{n} = \widehat{H}_{n}1_{B^{D \times \Sigma}}$. Noting \eqref{defi:local}, we obtain
	$$ \psi: = \sup_{n \in \mathbb{N}}  \sup_{0 \le t < \infty }\left| \int_{0}^{t} \widehat{H}_{n,u} dS_u \right|  < \infty, \ a.s.$$

From \eqref{eq:prob_conv}, the sequence of random variables
$\int_{\tau }^{t} \widehat{H}_{n,u}  1_{]]\tau,\zeta \wedge T_k]]}  dS_u, \ n \in \mathbb{N}$
converges uniformly in $t$ to a limit in the $L^0$ space as $n$ goes to infinity \footnote{The space $L^0$ equipped with the metric $d(f,g) = E[1 \wedge |f-g|]$ is a Fr\'echet space. The metric induces the topology of convergence in probability.}. Using Lemma A.1 of \cite{chau2020}, up to a subsequence, we may assume that the convergence is almost surely and uniformly in $t$.  

Using the density $dQ/dP= e^{-\psi}/E^P[ e^{-\psi}]$, we change to measure $Q$ equivalent to $P$  such that $\psi \in L^2(Q)$. Since the condition NUPBR holds globally, by repeating Lemmas 2.8,  3.3 of \cite{kabanov}, there exists $\overline{H}_n \in conv\{\widehat{H}_n, \widehat{H}_{n+1}, ... \}$ such that $\overline{H}_n \cdot S$ converges  in the semimartingale topology. Lemma \ref{lem:semi_closed} implies there is $H^*$ such that $\overline{H}_n \cdot S \to H^* \cdot S$ in the semimartingale topology. 	

Because the almost sure  convergence of $\int_{\tau }^{t} \widehat{H}_{n,u}  1_{]]\tau,\zeta \wedge T_k]]}  dS_u, \ n \in \mathbb{N}$ is uniform in $t$, we get the same property for the sequence of convex combinations
\begin{equation}\label{eq:sup_conv}
\sup_{t\ge 0}\left|  \int_{\tau}^{t }\overline{H}_{n,u}1_{]]\tau, T_k \wedge \zeta]]}dS_u - \int_{\tau}^{t }H^*_u1_{]]\tau, T_k \wedge \zeta]]}dS_u\right|  \to  0, \ a.s.,
\end{equation}
for any $k\in D, \zeta \in \Sigma$. Here, it's important to note that \eqref{eq:sup_conv} is stronger than the uniform convergence on compact  in probability obtained from the convergence of $\overline{H}$ in the semimartingale topology because $T_k \wedge \zeta$ can be infinity.  We compute for all $k \in D, \zeta \in \Sigma$ that
\begin{eqnarray*}
	&&\lim_{n \to \infty}\lim_{t \to \infty} \left|  \int_{\tau}^{t }\overline{H}_{n,u}1_{]]\tau, T_k \wedge \zeta]]}dS_u -  \int_{\tau}^{t }H^*_{u}1_{]]\tau, T_k \wedge \zeta]]}dS_u \right| \\
	&&  \le \lim_{n \to \infty} \sup_{ \tau \le t < \infty} \left| \int_{\tau}^{t }\overline{H}_{n,u}1_{]]\tau, T_k \wedge \zeta]]}dS_u - \int_{\tau}^{t }H^*_{u}1_{]]\tau, T_k \wedge \zeta]]}dS_u \right| \to  0, \ a.s.
\end{eqnarray*}
Therefore,  for all $k \in D, \zeta \in \Sigma$ 
\begin{eqnarray*}
	\lim_{t \to \infty}   \int_{\tau}^{t }H^*_{u}1_{]]\tau, T_k \wedge \zeta]]}dS_u &=& \lim_{n \to \infty} \lim_{t \to \infty} \int_{\tau}^{t }\overline{H}_{n,u}1_{]]\tau, T_k \wedge \zeta]]}dS_u \\
	&=& \lim_{n \to \infty}  \int_{\tau}^{\infty}\overline{H}_{n,u}1_{]]\tau, T_k \wedge \zeta]]}dS_u \\
	&=& g^{k,\zeta}_0 - v^{k,\tau,\zeta},\ a.s.,
\end{eqnarray*}
where the last equality comes from \eqref{eq:hat_H}.  In other words,  for all $k \in D, \zeta \in \Sigma$,
$$v^{k,\tau,\zeta} + \int_{\tau}^{\infty}H^*_u 1_{]]\tau,\zeta \wedge T_k]]}  dS_{u} = g^{k,\zeta}_0 \ge f^{k,\zeta}, \ a.s.$$
We get that $H^* 1_{]]\tau,\zeta \wedge T_k]]}  \in  \mathcal{A}_1(\tau  ,v^{k,\tau,\zeta})$ for all $k \in D, \zeta \in \Sigma$. The proof is complete. \qed
\end{proof}
Recall from \eqref{eq:tilde_C_0} that 
\begin{eqnarray*}
	\widetilde{\bC}^{D \times \Gamma \times \Sigma}_0 =\left\lbrace \left( \int_{0}^{T_k \wedge \zeta} \widetilde{H}_u(\theta) dS_u - h^{k,\theta,\zeta}\right)_{k \in D,\theta \in \Gamma, \zeta \in \Sigma}: \widetilde{H} \in \widetilde{\mathcal{A}}^{\Theta}, h^{k,\theta,\zeta} \in L^{0}_+ \right\rbrace.
\end{eqnarray*}
We also prove that the set $\widetilde{\bC}^{D \times \Gamma \times \Sigma}_0$ constructed from generalized strategies is Fatou-closed by using the arguments in the proof of Proposition \ref{pro:countable_Fatou}.
\begin{proposition}\label{pro:countable_Fatou_A}
	Let Assumption \ref{assumption_main}(i) hold. For any nonempty (possibly countable) set $D \times \Gamma \times \Sigma \subset \mathbb{N} \times \Theta \times \Theta$ such that $\infty \in \Gamma, \infty \in \Sigma$ and $x>0$, the set $\widetilde{\bC}^{D \times \Gamma \times  \Sigma}_0 $ is Fatou-closed, that is $\widetilde{\bC}^{D \times \Gamma \times  \Sigma}_0 \cap \prod_{k \in D, \theta \in \Gamma, \zeta \in \Sigma} B^{\infty}_{x}$ is closed in the space $\bL^{0,D \times \Gamma \times \Sigma}$,  for each $x \ge 0$.
\end{proposition}
\begin{proof}
	Without loss of generality, we may assume that $x = 1$. Let $\bff_n, n \in \mathbb{N}$ be a sequence in $\widetilde{\bC}^{D \times \Gamma \times  \Sigma}_0$ such that for every $k \in D, \theta \in \Gamma, \zeta \in \Sigma$ $$f^{k,\theta, \zeta}_n \to f^{k,\theta, \zeta} \text{ in probability and } f^{k,\theta, \zeta}_n \ge - 1,\ a.s.$$ We need to find $\widetilde{H} \in \widetilde{\mathcal{A}}^{\Theta}$  satisfying 
	\begin{equation}\label{eq:goal}
	\int_{0}^{T_k \wedge \zeta} \widetilde{H}_u(\theta) dS_u \ge f^{k,\theta, \zeta}, \ a.s., \forall k \in D, \theta \in \Gamma, \zeta \in \Sigma.
	\end{equation}
By Lemma A.1 of \cite{chau2020}, we may assume (up a subsequence) that $f^{k,\theta, \zeta}_n \to f^{k,\theta, \zeta}, \ a.s.$ for all $k \in D, \theta \in \Gamma, \zeta \in \Sigma$. By definition, there are $\widetilde{H}_n = (H_n,(H_n(\theta))_{\theta \in \Theta})\in \widetilde{\mathcal{A}}^{\Theta}, n \in \mathbb{N}$ such that for any $k \in D, \theta \in \Gamma, \zeta \in \Sigma$,
	\begin{equation}\label{eq:1_A}
		\int_{0}^{T_k \wedge \zeta} \widetilde{H}_{n,u}(\theta) dS_u \ge f^{k,\theta, \zeta}_n \ge -1, \ a.s.
	\end{equation}
Choosing $\theta = \infty, \zeta = \infty$ in \eqref{eq:1_A}, we get 
	\begin{equation}\label{eq:1_A1}
		\int_{0}^{T_k} H_{n,u} dS_u \ge -1, \ a.s. \ \forall k \in D,
	\end{equation}
Define $T_{sup} = \sup\{T_k, k \in D\}$ and note that $T_{sup} = \infty$ if $D$ is infinite. Since the condition NFLVR holds up to $T_{k}, k \in \mathbb{N}$, we redefine $H_n$ such that $H_n = H_n1_{[[0,T_{sup}]]} \in \mathcal{A}_1$, see Proposition 9.3.6 of \cite{delbaen2006mathematics}. For $\theta \in \Gamma \setminus \{\infty\}, \zeta = \infty$, \eqref{eq:1_A} yields
	\begin{equation}\label{eq:1_A2}
		\int_{0}^{T_k} \widetilde{H}_{n,u}(\theta) dS_u \ge -1, \ a.s. \ \forall k \in D.
	\end{equation}
From \eqref{eq:1_A2}, we redefine  $$H_{n,u}(\theta):= H_{n,u}(\theta)1_{]]\theta, T_{sup}]]}(u) \in \mathcal{A}_1\left(\theta , \int_0^{\theta } H_{n,u}dS_u \right).$$
For $\theta \notin \Gamma$, we simply redefine $H_{n,u}(\theta) = 0.$ The redefined sequence 
	$\widetilde{H}_n = (H_n,(H_n(\theta))_{\theta \in \Theta}), n \in \mathbb{N}$ is now in $ \widetilde{\mathcal{A}}^{\Theta}_1$ and satisfies \eqref{eq:1_A}. 
	
We denote
	\begin{eqnarray*}	&&\mathfrak{D}^{D\times \Gamma \times \Sigma} := \left\{ (g^{k,\theta, \zeta})_{k \in D, \theta \in \Gamma, \zeta \in \Sigma} :  \exists \widetilde{H}_n \in  \widetilde{\mathcal{A}}^{\Theta}_1, n \in \mathbb{N} \text{ such that }  \right.\\
	&&  \left. \int_{0}^{T_k \wedge \zeta} \widetilde{H}_{n,u}(\theta) dS_u \to  g^{k,\theta,\zeta}, \ a.s. \text{  and } g^{k,\theta,\zeta} \ge f^{k,\theta, \zeta}, \forall k \in D, \theta \in \Gamma, \zeta \in \Sigma \right\}.
	\end{eqnarray*}
	Noting that $\mathbb{N} \times \Theta \times \Sigma$ is countable, Lemma 4.4 of \cite{cfr} implies that there exists a sequence of convex combinations $\widetilde{H}^{\sharp}_n \in conv\{ \widetilde{H}_{n}, \widetilde{H}_{n+1},... \}$ such that  %$\widetilde{H}^{\sharp}_n=(H^{\sharp}_n,(H^{\sharp}_n(\theta))_{\theta \in \Gamma})$ such that 
	%\begin{eqnarray*}
%		\widetilde{H}^{\sharp}_{n,u}(\theta) &\in& conv\{\widetilde{H}_{n,u}(\theta),\widetilde{H}_{n+1,u}(\theta),...\} \\
	%	&=& conv\{H_{n,u}1_{u \le \theta} +H_{n,u}(\theta)1_{u > \theta} ,H_{n+1,u}1_{u \le \theta} +H_{n+1,u}(\theta)1_{u > \theta},...\}
	%\end{eqnarray*} for all $\theta \in \Gamma$ and
	 $$\int_{0}^{T_k \wedge \zeta}\widetilde{H}^{\sharp}_n(\theta) dS^{}_u \to q^{k,\theta, \zeta} \ge f^{k,\theta, \zeta}, \ a.s.$$ for some $q^{k,\theta, \zeta}$ taking finite values, for all $(k,\theta, \zeta) \in  D \times \Gamma \times \Sigma$. So the set $\mathfrak{D}^{D\times \Gamma \times \Sigma}$ is non-empty. It is also closed in $\bL^{0,D\times \Gamma  \times \Sigma}$. By Lemma \ref{lemma_maximal}, the set  $\mathfrak{D}^{D\times \Gamma \times \Sigma}$ has a maximal element $\mathbf{g}_0$. %It remains to check that there exists $\widetilde{H}^* = (H^*,(H^{*}(\theta))_{\theta \in \Gamma}) \in \widetilde{\mathcal{A}}^{\Theta}_1$ such that \begin{equation}\label{eq:goal}\int_{0}^{T_k \wedge \zeta} \widetilde{H}^*_{u}(\theta) dS_u  = g^{k,\theta, \zeta}_0, \ a.s.	\end{equation}for all $(k,\theta, \zeta) \in  D \times \Gamma \times \Sigma$.  
	 Hence, there exists a sequence $\widetilde{G}^{}_n = (G_n,G_n(\theta)) \in \widetilde{\mathcal{A}}^{\Theta}_1$ satisfying 
	\begin{equation}\label{eq:G}
		\int_{0}^{T_k \wedge \zeta} \widetilde{G}_{n,u}(\theta) dS_u \to  g^{k,\theta, \zeta}_0, \ a.s., \qquad \text{ and } \qquad  g^{k,\theta, \zeta}_0 \ge f^{k,\theta,\zeta}, \ a.s.
	\end{equation}for all $k \in D, \theta \in \Gamma, \zeta \in \Sigma$. %Define 
%\begin{eqnarray}
%	B^{D  \times  \Sigma}_{\le} &=& \bigcup_{k \in D, \zeta \in \Sigma} [[0, T_k \wedge \zeta]], \label{eq:Ble}\\
%	B^{D \times  \Sigma}_{\ge \theta} &=& \bigcup_{k \in D, \zeta \in \Sigma} ]]\theta,  T_k \wedge \zeta]], \ \forall \theta \in \Gamma, \label{eq:Bge}
%\end{eqnarray}the events on which our computation before and after the exercise times are carried out.  
For simplicity, we  redefine 
\begin{eqnarray}
G_n &=& G_n1_{[[0,T_{sup}]]}, \nonumber\\
G_n(\theta) &=& G_n(\theta)1_{]]\theta, T_{sup}]]}, \ \forall \theta \in \Gamma, \nonumber \\
G_n(\theta) &=& 0, \ \theta \notin \Gamma,  \label{eq:redef}
\end{eqnarray}
and the redefined sequence $\widetilde{G}^{}_n = (G_n,G_n(\theta)) \in \widetilde{\mathcal{A}}^{\Theta}_1$, too. 
We fix arbitrarily $k_0 \in \mathbb{N}, \zeta_0 \in \Sigma$ and claim that for any $\varepsilon > 0$
	\begin{equation}\label{eq:P-conv}
		\lim_{n,m \to \infty}P\left( \sup_{0 \le t}  \left| \int_{0}^{t} G_{n,u}1_{[[0,T_{k_0} \wedge \zeta_0]]} dS_u - \int_{0}^{t} G_{m,u} 1_{[[0,T_{k_0} \wedge \zeta_0]]} dS_u\right| > \varepsilon\right) = 0.	
	\end{equation}
If the claim \eqref{eq:P-conv} fails, there exist $n_{\ell},m_{\ell} \to \infty$ such that
	$$ P\left( \sup_{0 \le t}  \left| \int_{0}^{t} G_{n_{\ell},u} 1_{[[0,T_{k_0} \wedge \zeta_0]]} dS_u - \int_{0}^{t} G_{m_{\ell},u}1_{[[0,T_{k_0} \wedge \zeta_0]]} dS_u\right| > \alpha  \right) \ge \alpha $$
	for some $\alpha > 0$. Define 
	$$ U_{\ell} = \inf \left\lbrace  t \ge 0:   \int_{0}^{t} G_{n_{\ell},u}dS_u - \int_{0}^{t} G_{m_{\ell},u} dS_u \ge \alpha \right\rbrace .$$
	Then $P(U_{\ell}  \le   T_{k_0} \wedge \zeta_0) \ge \alpha$,  for every $\ell \in \mathbb{N}$. Define for all $\theta \in \Gamma$ the new strategy 
	\begin{eqnarray}\label{eq:switch}
		J_{\ell,u} &=&  1_{u \le U_{\ell}}G_{n_{\ell},u}+ 1_{u > U_{\ell}}  G_{m_{\ell},u}, \label{case1}\\	
		J_{\ell,u}(\theta)1_{u > \theta} &=& 1_{u > \theta} \left( 1_{\theta \le U_{\ell}}G_{n_{\ell},u}(\theta) + 1_{\theta > U_{\ell}}G_{m_{\ell},u}(\theta)\right), \label{eq:J}\\
		\widetilde{J}_{\ell,u}(\theta) &=& J_{\ell,u}1_{u \le  \theta} + J_{{\ell},u}(\theta)1_{u > \theta}.\nonumber
	\end{eqnarray}
The new strategy $\widetilde{J}^{}_{\ell} = (J_{\ell},J_{\ell}(\theta)) \in \widetilde{\mathcal{A}}^{\Theta}_1$ by Lemma 2.3 (b) of \cite{kabanov}. 
Here, we only switch from $G_{n_{\ell}}$ to $G_{m_{\ell}}$ at $U_{\ell}$, and if $U_{\ell} < \theta$, and we have to switch $G_{n_{\ell}}(\theta)$ to $G_{m_{\ell}}(\theta)$ as a consequence of the first switching, see \eqref{eq:J}. If $U_{\ell} \ge \theta$, we keep using $G_{n_{\ell}}(\theta)$, as in \eqref{eq:J}, because the switch from  $G_{n_{\ell}}(\theta)$ to $G_{m_{\ell}}(\theta)$ may not improve the final payoffs.  For all $k \in D, \theta \in \Gamma, \zeta \in \Sigma$, we consider the following cases:
	\begin{itemize}
		\item[(i)]  The case $ T_k \wedge \zeta \le \theta$: the option is not exercised before $T_k \wedge \zeta$, we switch from the strategy  $G_{n_{\ell}}$ to the strategy $G_{m_{\ell}}$ at the stopping time $U_{\ell}$ as in \eqref{case1} and obtain that 
		\begin{eqnarray*}
			&&	\int_{0}^{ T_k \wedge \zeta}\widetilde{J}_{\ell,u}(\theta) dS_{u} \nonumber \\
			&=& 1_{U_{\ell} >  T_k \wedge \zeta} \int_{0}^{ T_k\wedge \zeta}G_{n_{\ell},u} dS_{u}\nonumber  \\
			&& + 1_{U_{\ell} \le T_k \wedge \zeta}\left( \int_{0}^{U_{\ell} }G_{n_{\ell},u} dS_{u} +  \int_{0}^{ T_k \wedge \zeta}G_{m_{\ell},u}  dS_{u} - \int_{0}^{U_{\ell} }G_{m_{\ell},u}  dS_{u} \right) \nonumber \\
			&=& 1_{U_{\ell} >  T_k \wedge \zeta} \int_{0}^{T_k \wedge \zeta}G_{n_{\ell},u}  dS_{u}  + 1_{U_{\ell} \le  T_k \wedge \zeta} \int_{0}^{T_k \wedge \zeta}G_{m_{\ell},u}  dS_{u} +  \xi^{k,\theta,\zeta}_{\ell}, \label{eq2}
		\end{eqnarray*}
		where we define $$\xi^{k,\theta,\zeta}_{\ell}
		= \left( \int_{0}^{U_{\ell} }G_{n_{\ell},u} dS_{u}  - \int_{0}^{U_{\ell} }G_{m_{\ell},u}dS_{u}\right)1_{U_{\ell} \le  T_k \wedge \zeta}.$$ 
		\item[(ii)] The case $\theta < T_k \wedge \zeta$ and $\theta \le U_{\ell}$: the option is exercised at time $\theta$ and there is no switching before $\theta$. In this case, we use the strategy in \eqref{eq:J} on the event $\{\theta \le U_{\ell}\}$ and obtain 
		\begin{eqnarray*}
			\int_{0}^{T_k \wedge \zeta}\widetilde{J}_{\ell,u}(\theta)dS_{u}
			&=& \int_{0}^{\theta} G_{n_{\ell},u} dS_u +  \int_{\theta }^{ T_k \wedge \zeta}  G_{n_{\ell},u}(\theta) dS_u \to g^{k,\theta, \zeta}. 
		\end{eqnarray*}
		\item[(iii)] The case $\theta <  T_k \wedge \zeta$ and $\theta > U_{\ell}$: the option is exercised at time $\theta$ and there is a  switching before $\theta$. In this case,  we use the strategy in \eqref{eq:J} on the event $\{\theta > U_{\ell}\}$, 
		\begin{eqnarray*}
			\int_{0}^{T_k \wedge \zeta}\widetilde{J}_{\ell,u}(\theta) dS_{u}
			&=& \int_{0}^{U_{\ell}} G_{n_{\ell},u} dS_u + \int_{U_{\ell}}^{\theta } G_{m_{\ell},u} dS_u +  \int_{\theta }^{T_k \wedge \zeta}  G_{m_{\ell},u}(\theta) dS_u \\
			&=& \int_{0}^{U_{\ell}} G_{n_{\ell},u} dS_u - \int_{0}^{U_{\ell}} G_{m_{\ell},u} dS_u + \int_{0}^{\theta} G_{m_{\ell},u} dS_u \\
			&& +  \int_{\theta }^{T_k \wedge \zeta}  G_{m_{\ell},u}(\theta) dS_u  \\
			&=& \xi^{k,\theta,\zeta}_{\ell}  + 
			\int_{0}^{T_k \wedge \zeta}  \widetilde{G}_{m_{\ell},u}(\theta) dS_u.
		\end{eqnarray*}
	\end{itemize}
For any $k \in D, \theta \in \Gamma, \zeta \in \Sigma,$ the final payoffs are improved in the cases $(i)$ and $(iii)$. For our argument, we choose $\theta = \infty$ and thus the cases $(ii)$ and $(iii)$ do not occur.  Therefore, for every ${\ell} \in \mathbb{N}$, 	$$P(\xi^{k_0,\infty,\zeta_0}_{\ell}\ge \alpha)  =  P(U_{\ell}  \le  T_{k_0}  \wedge \zeta_0) \ge \alpha.$$
By Lemma 9.8.1 of \cite{delbaen2006mathematics} and Lemma 4.4 of \cite{cfr}, there exist convex combinations
	\begin{eqnarray*}
		\xi^{k,\theta,\zeta,\sharp}_{\ell} &\in& conv\{\xi^{k,\theta,\zeta}_{\ell}, \xi^{k,\theta,\zeta}_{\ell+1},...\}
	\end{eqnarray*} 
	such that $\xi^{k,\theta,\zeta,\sharp}_{\ell} \to \eta^{k, \theta, \zeta}$ with $\eta^{k, \theta, \zeta} \ge 0$ for all $k \in D, \theta \in \Gamma, \zeta \in \Sigma$ and in addition, we have $\eta^{k_0, \infty, \zeta_0} \ne 0$. Since $\eta^{k_0, \infty, \zeta_0} \ne 0$, the corresponding sequence of convex combinations for $\widetilde{J}_{\ell}$ converges to an element larger than $\mathbf{g}_0$. Therefore, $\mathbf{g}_0$ is not a maximal element of $\mathfrak{D}^{D \times \Gamma \times \Sigma}$, which is a contradiction.  
	
We conclude that the claim \eqref{eq:P-conv} holds true and hence the random variable $$  \sup_{n \in \mathbb{N}} \sup_{0 \le t }\left| \int_{0}^{t} G_{n,u}1_{[[0,T_{k} \wedge \zeta ]]}dS_u \right|  < \infty, \ a.s.$$
for any $k \in D, \zeta \in \Sigma$. Therefore,
	$$\psi = \sup_{n \in \mathbb{N}} \sup_{0 \le t  }\left| \int_{0}^{t} G_{n,u}dS_u \right| < \infty $$
	because  of the redefined sequence $G_n$, see  \eqref{eq:redef}. We change to measure $Q$ using the density $dQ/dP = e^{-\psi}/E^P[ e^{-\psi}]$ such that $\psi \in L^2(Q)$. Since NUPBR holds globally, Lemmas 2.8, 2.9 of \cite{kabanov} imply that there exists $G^{\sharp}_n \in conv\{{G}_n, {G}_{n+1}, ... \}$ such that $G^{\sharp}_n \cdot S$ converges in the semimartingale topology. Lemma \ref{lem:semi_closed} implies there is $H^*$ such that $G^{\sharp}_n \cdot S \to {H}^* \cdot S$ in the the semimartingale topology. Since $G_n \in \A_1$, we obtain that $H^* \in \mathcal{A}_1.$
	
Using the same weights in the construction of $G^{\sharp}_n$, we obtain $G^{\sharp}_n(\theta)$. From \eqref{eq:G}, for any $k \in D, \theta \in \Gamma, \zeta \in \Sigma$ and $\theta < \infty$, 
\begin{equation}\label{eq:G_sharp}
\int_{0}^{\theta}  G^{\sharp}_n1_{[[0, T_k \wedge \zeta]]} dS_u + \int_{\theta}^{\infty} G^{\sharp}_n(\theta)1_{]]\theta, T_k \wedge \zeta]]}dS_u\to g^{k,\theta, \zeta}_0 \ge f^{k,\theta, \zeta}, \ a.s
\end{equation}
as $n \to \infty$. 	We define $v^{k,\theta,\zeta} := \int_{0}^{\theta}  H^*_u1_{[[0, T_k \wedge \zeta]]}dS_u$. Because convergence in the semimartingale topology implies uniform convergence on compacts in probability, 
	\begin{eqnarray*}
	\left|  \int_0^{\theta  \wedge T_k \wedge \zeta} G^{\sharp}_{n,u}dS_u  -  \int_{0}^{\theta  \wedge T_k \wedge \zeta}  H^*_udS_u\right|  &\le& \sup_{s \le \theta} \left|  \int_0^{s} G^{\sharp}_{n,u}dS_u  -  \int_{0}^{s}  H^*_udS_u\right| \\
	 &\to& 0
	\end{eqnarray*}
	in probability for all $k \in D, \theta \in \Gamma, \zeta \in \Sigma$ and $\theta < \infty$.  
	Lemma A.1 of \cite{chau2020} implies that up to a subsequence
	$$ \int_0^{\theta } G^{\sharp}_{n,u}1_{[[0, T_k \wedge \zeta]]}dS_u \to v^{k,\theta,\zeta} , \ a.s.$$
	for all $k \in D, \theta \in \Gamma,  \zeta \in \Sigma$ and $\theta < \infty$. Thus  \eqref{eq:G_sharp} yields
	$$v^{k,\theta,\zeta} + \int_{\theta}^{\infty}  G^{\sharp}_n(\theta)1_{]]\theta, T_k \wedge \zeta]]} dS_u\to g^{k,\theta,\zeta}_0 \ge f^{k,\theta,\zeta}, \ a.s.$$
	For each $\theta \in \Gamma$ and $\theta < \infty$, applying Proposition \ref{pro:countable_Fatou} to the set $\mathbf{C}^{D \times \Sigma}(\theta, \bv)$ where $\bv = (v^{k,\theta,\zeta})_{k \in D, \zeta \in \Sigma}$ yields the existence of $H^*(\theta)$ such that for all $k \in D, \zeta \in \Sigma$, 
\begin{eqnarray*}
v^{k,\theta,\zeta} + \int_{\theta }^{\infty}   H^*(\theta)1_{]]\theta,  T_k \wedge \zeta]]} dS_u &\ge& g^{k,\theta, \zeta}_0 \ge  f^{k,\theta, \zeta}, \ a.s.\\
H^*(\theta)1_{]]\theta, T_k \wedge  \zeta]]} &\in&  \mathcal{A}_1(\theta ,v^{k,\theta,\zeta}).	
\end{eqnarray*}
We redefine $H^*(\theta) = H^*(\theta)1_{]]\theta, T_{sup}]]}$. The redefined strategy satisfies  $H^*(\theta) \in \A_1\left( \theta, \int_0^{\theta}H^*_udS_u\right)$. Setting $H^*(\theta) = 0$ for $\theta \notin \Gamma$, the strategy $(H^*,(H^*(\theta)_{\theta \in \Theta}) \in \widetilde{\mathcal{A}}^{\Theta}_1$ satisfies \eqref{eq:goal}. The proof is complete. \qed
\end{proof}
\subsection{Proof of Theorem \ref{thm:local_closed_0}}\label{sec:proof_thm:local_closed_0}
	Let $\bff_{\alpha}, \alpha \in I$ be a net in $\widetilde{\bC}^{\mathbb{N} \times \Theta \times \Theta}$, i.e., 
	$$ \bff_{\alpha} \le  \left(  \int_{0}^{T_k \wedge \zeta} \widetilde{H}_{\alpha,u}(\theta) dS_u \right)_{k \in \mathbb{N}, \theta \in \Theta, \zeta \in \Theta}, \text{ for some } \widetilde{H}_{\alpha} \in \widetilde{\A}^{\Theta},$$ such that $\bff_{\alpha} \to \bff$ in $\bL^{\infty,\mathbb{N} \times \Theta \times \Theta,b}$ in the $\bw^*$ topology. We need to prove that $\bff \in \widetilde{\bC}^{\mathbb{N} \times \Theta \times \Theta}$, that is there exists $\widetilde{H} \in \widetilde{\mathcal{A}}^{\Theta}$ satisfying  
	$$\int_{0}^{T_k \wedge \zeta}\widetilde{H}_{u}(\theta) dS_{u} \ge f^{k,\theta, \zeta}, \ a.s. \ \forall (k,\theta, \zeta) \in \mathbb{N} \times \Theta \times \Theta.$$
	By the definition of $\bff$, there exists $x \in \mathbb{R}$ such that $f^{k,\theta, \zeta} \ge -x$ for all $(k,\theta, \zeta) \in \mathbb{N} \times \Theta \times \Theta$. For each $(k,\theta, \zeta) \in \mathbb{N} \times \Theta \times \Theta$, we define 
	\begin{eqnarray*}
	\mathcal{H}^{k,\theta, \zeta} &=& \left\lbrace \left( \int_{0}^{T_{k'} \wedge \zeta'}\widetilde{H}_{u}(\theta') dS_{u}\right)_{k' \in \mathbb{N},\theta' \in \Theta, \zeta' \in \Theta} : \widetilde{H} \in \widetilde{\A}^{\Theta}_x \right.  \\
	&& \qquad \qquad \left. \text{ and }  
	\int_{0}^{T_k \wedge \zeta }\widetilde{H}_{u}(\theta) dS_{u} \ge f^{k,\theta, \zeta}, \ a.s. \right\rbrace .
	\end{eqnarray*}
	The set $\mathcal{H}^{k,\theta, \zeta}$ is clearly closed and convex subset of  $\bL^{0,\mathbb{N} \times \Theta \times \Theta}$.  To complete the proof, we need to prove
	\begin{equation}\label{eq:cap_A}
		\bigcap_{k \in \mathbb{N},\theta \in \Theta, \zeta \in \Theta} \H^{k,\theta, \zeta} \ne \emptyset. 
	\end{equation}
	First, we will prove that 
	\begin{equation}\label{eq:finite_intersec_A}
		\mathcal{H}^{D \times \Gamma \times \Sigma}=\bigcap_{(k,\theta, \zeta) \in D \times \Gamma \times \Sigma} \H^{k,\theta, \zeta} \ne \emptyset,
	\end{equation}
	where $D \times \Gamma \times \Sigma$ is an arbitrary nonempty finite subset of $\mathbb{N} \times \Theta \times \Theta$ such that $\infty \in \Gamma$ and $\infty \in \Sigma$. Proposition \ref{pro:countable_Fatou_A} implies that the set $\widetilde{\bC}_0^{D \times \Gamma \times \Sigma}$ is Fatou-closed. Therefore, using Proposition \ref{cfr}, the set $\widetilde{\bC}_0^{D \times \Gamma \times \Sigma}$ is $\bw^*$-closed in $\bL^{\infty,D \times \Gamma \times \Sigma,b}$. Since $f^{k,\theta, \zeta}_{\alpha} \to f^{k,\theta, \zeta}$ in the $w^*$ topology for each $(k,\theta, \zeta) \in D \times \Gamma \times \Sigma$, and we obtain that $(f^{k,\theta, \zeta})_{(k,\theta, \zeta) \in D\times \Gamma \times \Sigma} \in \widetilde{\bC}_0^{D \times \Gamma \times \Sigma}$. Therefore, we can find $\widetilde{H}^{D \times \Gamma \times \Sigma} = (H^{D \times \Gamma \times \Sigma}, (H^{D \times \Gamma \times \Sigma}(\theta))_{\theta \in \Theta})\in \widetilde{\mathcal{A}}^{\Theta}$ satisfying
	$$\int_{0}^{T_k \wedge \zeta}\widetilde{H}^{D \times \Gamma \times \Sigma}_{u}(\theta) dS_{u} \ge f^{k,\theta, \zeta} \ge -x, \ a.s. \ \forall (k,\theta, \zeta) \in D \times \Gamma \times \Sigma.$$
Repeating the argument involving \eqref{eq:1_A}, \eqref{eq:1_A1}, \eqref{eq:1_A2} in  the proof of Proposition \ref{pro:countable_Fatou_A},  we redefine $\widetilde{H}^{D \times \Gamma \times \Sigma}$ such that $\widetilde{H}^{D \times \Gamma \times \Sigma} \in \widetilde{ \mathcal{A}}^{\Theta}_x$ and thus, (\ref{eq:finite_intersec_A}) holds true. 

Denote by
$$\mathcal{H}^{\mathbb{N} \times \Theta \times \Theta} = \left\lbrace  \left( \int_0^{T_k \wedge \zeta} \widetilde{H}_u(\theta) dS^{}_u\right) _{k \in \mathbb{N},\theta \in \Theta, \zeta \in \Theta}: \widetilde{H} \in \widetilde{\A}^{\Theta}_x  \right\rbrace $$
	and by $\overline{\mathcal{H}^{\mathbb{N} \times \Theta \times \Theta}}$ the closure of $\mathcal{H}^{\mathbb{N} \times \Theta \times \Theta}$ in $(L^0)^{\mathbb{N} \times \Theta \times \Theta}$. 
	It is clear that  $\overline{\mathcal{H}^{\mathbb{N} \times \Theta \times \Theta}} $ is a convex set. Furthermore, $\mathcal{H}^{\mathbb{N} \times \Theta \times \Theta}$ is $c$-bounded by Lemma \ref{lemma:c-bounded}, and so is $\overline{\mathcal{H}^{\mathbb{N} \times \Theta \times \Theta}}  \subset (L^0)^{\mathbb{N} \times \Theta \times \Theta}$. Noting that $\mathbb{N} \times \Theta \times \Theta$ is countable, the set $\overline{\mathcal{H}^{\mathbb{N} \times \Theta \times \Theta}} $ is convexly compact by Proposition \ref{bfl}. Since $\mathcal{H}^{k,\theta, \zeta} \subset \overline{\mathcal{H}^{\mathbb{N} \times \Theta \times \Theta}}$ for any $k \in \mathbb{N}, \theta \in \Theta, \zeta \in \Theta$, we conclude that (\ref{eq:cap_A}) holds. The proof is complete. \qed 
\begin{remark}
		In \cite{cfr}, under the NUPBR condition,  trading strategies are bounded in total variation norm because of trading costs and it is possible to exploit the convex compactness property of the set of admissible trading strategies. As a consequence, their techniques are simpler. In the present frictionless settings, it is impossible to have such boundedness for trading strategies, and therefore we have to work with wealth processes, as in the proof of Theorem \ref{thm:local_closed_0}. 
\end{remark}
\subsection{Proof of Proposition \ref{pro:1}}
Let $D = \{k_1 < ...< k_p\}$,  $\Gamma = \{\theta_1 < ... < \theta_q\}$, $\Sigma = \{ \zeta_1, < ... < \zeta_r\}$, and recall the stopped process from \eqref{eq:stopped_wealth} $$\widetilde{W}^{k,\theta,\zeta}_t(\widetilde{H})  =\int_0^{t} \widetilde{H}_u(\theta)dS^{T_k \wedge \zeta}_u.$$
Using It\^o's formula, we obtain
\begin{eqnarray*}
&&d(Z^{k,\theta,\zeta}_t\widetilde{W}^{k,\theta,\zeta}_t(\widetilde{H}) ) \\
&=& \widetilde{W}^{k,\theta,\zeta}_{t-}(\widetilde{H})  dZ^{k,\theta,\zeta}_t + Z^{k,\theta,\zeta}_{t-}d\widetilde{W}^{k,\theta,\zeta}_t(\widetilde{H}) + d[Z^{k,\theta,\zeta},\widetilde{W}^{k,\theta,\zeta}(\widetilde{H})]_t \\
		&=& \widetilde{W}^{k,\theta,\zeta}_{t-}(\widetilde{H}) dZ^{k,\theta,\zeta}_t + Z^{k,\theta,\zeta}_{t-}\left(1_{t \le \theta} H_tdS^{T_k \wedge \zeta}_{t} + 1_{t > \theta} H_t(\theta)dS^{T_k \wedge \zeta}_{t} \right)  \\
		&&\qquad  +  \left(1_{t \le \theta} H_t d[Z^{k,\theta,\zeta},S^{T_k \wedge \zeta}]_{t} + 1_{t > \theta} H_t(\theta) d[Z^{k,\theta,\zeta},S^{T_k \wedge \zeta}]_{t} \right) . 
	\end{eqnarray*}
Therefore, 
	\begin{eqnarray}
	&&\sum_{(k,\theta,\zeta) \in D \times \Gamma \times \Sigma}	Z^{k,\theta,\zeta}\widetilde{W}^{k,\theta,\zeta}_{\infty}(\widetilde{H}) \nonumber \\
	&=& \sum_{(k,\theta,\zeta) \in D \times \Gamma \times \Sigma} \int_{0}^{\infty} \widetilde{W}^{k,\theta,\zeta}_{t-}(\widetilde{H}) dZ^{k,\theta,\zeta}_t \nonumber \\
	&& \qquad+  \sum_{(k,\theta) \in D \times \Gamma \times \Sigma}\int_0^{\theta} Z^{k,\theta,\zeta}_{t-} H_tdS^{T_k \wedge \zeta}_{t} +  H_t d[Z^{k,\theta,\zeta},S^{T_k \wedge \zeta}]_{t} \nonumber \\
	&& \qquad + \sum_{(k,\theta, \zeta) \in D \times \Gamma \times \Sigma} \int_{\theta}^{T_k \wedge \zeta} Z^{k,\theta,\zeta}_{t-} H_t(\theta)dS^{T_k \wedge \zeta}_{t} +  H_t(\theta) d[Z^{k,\theta,\zeta},S^{T_k \wedge \zeta}]_{t}  \nonumber \\
		&=& \sum_{(k,\theta, \zeta) \in D \times \Gamma \times \Sigma} \int_{0}^{\infty} \widetilde{W}^{k,\theta,\zeta}_{t-}(\widetilde{H}) dZ^{k,\theta,\zeta}_t \nonumber\\
		&& \qquad+  \sum_{(k,\theta, \zeta) \in D \times \Gamma \times \Sigma}\int_0^{\theta } H_t \left( d(Z^{k,\theta, \zeta}_tS^{T_k \wedge \zeta}_{t}) - S^{T_k \wedge \zeta}_{t-}dZ^{k,\theta,\zeta}_t\right)  \nonumber \\
		&& \qquad+ \sum_{(k,\theta, \zeta) \in D \times \Gamma \times \Sigma} \int_{\theta}^{\infty}  H_t(\theta)\left( d(Z^{k,\theta,\zeta}_tS^{T_k \wedge \zeta}_{t}) - S^{T_k \wedge \zeta}_{t-}dZ^{k,\theta,\zeta}_t\right). \label{eq:deflated}
	\end{eqnarray}
We consider the case $\mathfrak{C} = \mathfrak{C}^u$. For any $\theta \in \Gamma$, the property \eqref{eq:martingale_after1} implies that $\sum_{k \in D, \zeta \in \Sigma}Z^{k, \theta, \zeta}\mathcal{S}^{1,T_{k} \wedge \zeta}$, ...,$\sum_{k \in D, \zeta \in \Sigma}Z^{k, \theta, \zeta}\mathcal{S}^{d,T_{k} \wedge \zeta}$ are martingales on $[\theta,\infty)$ (see also Remark \eqref{ex:density}), therefore, 
		$$ \int_{\theta}^{t}  H_u(\theta) d\left( \sum_{k \in D, \zeta \in \Sigma} Z^{k,\theta}_uS^{T_k \wedge \zeta}_{u}\right) $$
		is a local martingale on $[\theta,\infty)$. 
		Similarly, the term
		\begin{eqnarray}
&&\sum_{(k,\theta, \zeta) \in D \times \Gamma \times \Sigma} \int_0^{t \wedge \theta} H_u d(Z^{k,\theta,\zeta}_uS^{T_k \wedge \zeta}_{u}) \nonumber \\
&& \qquad = \int_0^{t \wedge\theta_1} H_u d \left(  \sum_{k \in D, \zeta \in \Sigma} \sum_{j=1}^{q} Z^{k,\theta_j, \zeta}_uS^{T_k \wedge \zeta}_{u}\right) 	\nonumber\\
&& \qquad + \int_{t \wedge\theta_1}^{t \wedge \theta_2} H_u d \left(  \sum_{k \in D, \zeta \in \Sigma} \sum_{j=2}^{q} Z^{k,\theta_j,\zeta}_uS^{T_k \wedge \zeta}_{u}\right) \nonumber\\
&&\qquad + ...+ \int_{t \wedge \theta_{q-1}}^{t \wedge \theta_q} H_u d \left(  \sum_{k \in D, \zeta \in \Sigma}  Z^{k,\theta_q,\zeta}_uS^{T_k \wedge \zeta}_{u}\right), \label{eq:term2}
		\end{eqnarray}
		is also a local martingale by \eqref{eq:mart_11}, see also Remark \eqref{ex:density}. Since $\widetilde{W}^{k,\theta, \zeta}_t(\widetilde{H})$ is uniformly bounded from below and $Z_t^{k,\theta, \zeta}, (k,\theta, \zeta) \in D \times \Gamma \times \Sigma$ are martingales, the last quantity  of (\ref{eq:deflated}) is bounded from below by a martingale and hence a supermartingale by Fatou's lemma. Therefore, we have 
		\begin{eqnarray*}
			E\left[ \sum_{(k,\theta, \zeta) \in D \times \Gamma \times \Sigma}	Z^{k,\theta,\zeta}\widetilde{W}^{k,\theta, \zeta}_{\infty}(\widetilde{H}) \right] \le 0.
		\end{eqnarray*}
We study the case $\mathfrak{C} = \mathfrak{C}^s$. This is the case where the first $n$ stocks ($1\le n \le  d$) are prohibited from short selling. Again, for any $\theta \in \Gamma$, the processes $\sum_{k \in D, \zeta \in \Sigma}Z^{k, \theta, \zeta}\mathcal{S}^{1,T_{k} \wedge \zeta}$, ...,$\sum_{k \in D, \zeta \in \Sigma}Z^{k, \theta, \zeta}\mathcal{S}^{n,T_{k} \wedge \zeta}$ are supermartingales, while $\sum_{k \in D, \zeta \in \Sigma}Z^{k, \theta, \zeta}\mathcal{S}^{n+1,T_{k}\wedge \zeta},$ ...,$\sum_{k \in D, \zeta \in \Sigma}Z^{k, \theta,\zeta }\mathcal{S}^{d,T_{k} \wedge \zeta}$ are martingales, see Remark \ref{ex:density_short}. Noting that $H^1 \ge 0,..,H^n \ge 0$, the vector stochastic integral $$ \int_{\theta}^{t}  H_u(\theta) d\left( \sum_{k \in D, \zeta \in \Sigma}Z^{k,\theta,\zeta}_uS^{T_k \wedge \zeta}_{u}\right) $$
is a local supermartingale on $[[\theta, \infty))$ by Lemma \ref{lem:int_super} below. Similarly, the term \eqref{eq:term2} is also a local supermartingale. Fatou's lemma yields the conclusion. The proof of Proposition \ref{pro:1} is complete. \qed 

\begin{lemma}\label{lem:int_super}
		Let $X = (X^1,...,X^d)$ be such that for some $1 \le n \le d$, the processes $X^1,...,X^n$ are supermartingales and $X^{n+1},...,X^d$ are martingales. Assume that $H = (H^1,...,H^d)$ is integrable w.r.t $X$ and $H^i \ge 0, i = 1,...,n$. Then the vector stochastic integral $\int H dX$ is a local supermartingale. 
	\end{lemma}
\begin{proof}
Define $H^n = H1_{\|H\| \le n}$ for  $n \in \mathbb{N}$. By Lemma 3.3 of \cite{pulido}, $H^n \cdot X$ is a local supermartingale. Therefore, $H \cdot X$ is also a local supermartingale. \qed
	
%			By the Doob–Meyer decomposition theorem, for $i=1,...,n$, there are unique martingale $M^i$ and predictable non-decreasing $A^i$ processes such that $X^i = M^i - A^i$. The $d$-dimensional process $\widetilde{M} = (M^1,...,M^n, X^{n+1},...,X^n)$ is a martingale and the process $\widetilde{A} = (A^1,...,A^n,0,...,0)$ is non-decreasing. We obtain that		$$\int H dX = \int H d\widetilde{M} - \int H d\widetilde{A}.$$		The term $\int H d\widetilde{A}$ is non-negative because $H^1,...,H^n$ are all non-negative. The term $\int H d\widetilde{M}$ is a local martingale, see \cite{sc2002}. Therefore, $\int H dX$ is a local supermartingale.  \qed 
\end{proof}

\subsection{Proof of Corollary \ref{cor:density_A} }
We note that the implication $(ii) \Rightarrow (iii)$ is obvious.

$(i) \Rightarrow (ii)$: Let us fix $D \times \Gamma \times \Sigma$ containing $(k,\theta, \zeta)$  such that $\infty \in \Gamma, \infty \in \Sigma$ and $A \in \mathcal{F}_{T_{k} \wedge \zeta}$ with $P(A) > 0$. By   Theorem \ref{thm:local_closed_0}, the convex set $\widetilde{\bC}^{D \times \Gamma \times \Sigma}$ is $\bw^*$-closed in $\bL^{\infty,D \times \Gamma \times \Sigma,b}$. Applying the Hahn-Banach theorem to the compact set $1_A\mathbf{1}^{k,\theta, \zeta}$ and the $\bw^*$-closed convex set $\widetilde{\bC}^{D \times  \Gamma \times \Sigma}$, there exists $\mathbf{Q} \in \left(  \bL^{\infty,D \times \Gamma \times \Sigma,b} \right) ^*$  such that 
$$  \sup_{\bff \in \widetilde{\bC}^{D\times \Gamma \times \Sigma}} \mathbf{Q}(\bff) \le \alpha < \beta \le \mathbf{Q}(1_A\mathbf{1}^{k,\theta, \zeta}).$$
By Theorem 3.6 of \cite{rudin}, $\mathbf{Q}$ has a continuous linear extension $\overline{\mathbf{Q}}$ to $\bL^{\infty,D \times \Gamma \times \Sigma}$ such that $\overline{\mathbf{Q}}|_{\bL^{\infty,D\times \Gamma \times \Sigma,b}} = \mathbf{Q} $. We identify	$$\overline{\mathbf{Q}} = (Z^{k,\theta,\zeta})_{k \in D, \theta \in \Gamma, \zeta \in \Sigma} \in \bigoplus_{k \in D, \theta \in \Gamma, \zeta \in \Sigma}  L^1(\mathcal{F}_{T_k \wedge \zeta},P),$$ and normalize  $\overline{\mathbf{Q}}(\mathbf{1})=1$. If there are no confusion, we write $\mathbf{Q}$ instead of $\overline{\mathbf{Q}}$. Since $\mathbf{0} \in \widetilde{\bC}^{D \times  \Gamma \times \Sigma}$, it follows that $\alpha \ge 0$. Since $\widetilde{\bC}^{D \times  \Gamma \times \Sigma}$ is a cone, we must have 
\begin{equation}\label{eq:Q_ro}
	\mathbf{Q}(\mathbf{f}) \le 0, \ \forall \mathbf{f} \in \widetilde{\bC}^{D \times  \Gamma \times \Sigma},
\end{equation}
and as a consequence,  $Z^{k,\theta, \zeta} \ge 0, \forall k \in D, \theta \in \Gamma, \zeta \in \Sigma$. Note that $E[Z^{k,\theta, \zeta}1_A] > 0.$ 

We assume that $D = \{k_1 < ...< k_p\}$,  $\Gamma = \{\theta_1 < ... < \theta_q\}$, and $\Sigma = \{ \zeta_1, < ... < \zeta_r\}$. Let $0 \ne u \in \mathfrak{C}$. For $\theta_{\ell} < \sigma \le \tau \le  \theta_{\ell+1}$ where $\ell \in \{0, 1,...,q-1\}, B \in \F_{\sigma}$ (with the convention $\theta_0 = 0$), the strategy $\widetilde{H} = (H, (H(\theta))_{\theta \in \Theta})$ with
$H = 1_{B} \frac{u}{||u||} 1_{]]\sigma, \tau]]}, H(\theta) = 0$ belongs to $\widetilde{\mathcal{A}}^{\Theta}$, noting that $S$ is bounded up to $T_{k_p}$. Therefore, \eqref{eq:stopped_wealth} implies that
$$W^{k,\theta, \zeta}_{\infty} = \int_0^{\theta} 1_B\frac{u}{\|u\|}1_{]]\sigma, \tau]]} dS^{T_{k} \wedge \zeta}_u, \ k \in D, \theta \in \Gamma, \zeta \in \Sigma$$ and then
\begin{eqnarray*}\label{eq:mart_2}
	E\left[ \sum_{i=1}^{p} \sum_{s = 1}^{r} \sum_{j=\ell+1}^{q}  Z^{k_i,\theta_j, \zeta_s} \left( u \cdot S^{T_{k_i} \wedge \zeta_s}_{\tau} - u \cdot S^{T_{k_i} \wedge \zeta_s}_{\sigma} \right)1_{B}  \right] \le 0,
\end{eqnarray*} 
and then \eqref{eq:mart_11} follows. Next, fixing $\theta_j, j \in \{1,...,q\}$ and choosing $\theta_j < \sigma \le \tau$,
\begin{eqnarray*}
	H_t = 0, \qquad H(\theta_j) = 1_B\frac{u}{||u||}1_{]]\sigma, \tau]]}, \qquad 	H(\theta) =  0, \ \forall \theta \ne \theta_j, 
\end{eqnarray*}
\eqref{eq:stopped_wealth} yields that
$$W^{k,\theta_j, \zeta}_{\infty} = \int_{\theta_j}^{\infty} 1_B\frac{u}{\|u\|}1_{]]\sigma, \tau]]} dS^{T_{k} \wedge \zeta}_u, \ k \in D, \zeta \in \Sigma,$$ 
and we obtain 
\begin{equation*}\label{eq:martingale_after}
	E\left[ \sum_{i=1}^{p} \sum_{s=1}^{r} Z^{k_i,\theta_j, \zeta_s} \left( u \cdot  S^{T_{k_i} \wedge \zeta_s}_{\tau} - u \cdot  S^{T_{k_i} \wedge \zeta_s}_{\sigma} \right)1_{B}  \right] \le 0.
\end{equation*}
and hence \eqref{eq:martingale_after1}. 

$(iii) \Rightarrow (i)$: Assume $\mathfrak{C} = \mathfrak{C}^u$. Choose $D = \{k\}$. There exists a $(D \times \{\infty\})$-pricing system $0 \le Z^{k,\infty}, E[Z^{k,\infty}] = 1$ such that $Z^{k,\infty}_tS^{T_k}_t$ is a local martingale on $[0,\infty)$.  Proposition \ref{pro:1} implies the condition NFLVR$_{\mathfrak{C}}$ holds for the market up to time $ T_k.$ The case $\mathfrak{C} = \mathfrak{C}^s$ is treated similarly.
\subsection{Proof of Proposition \ref{pro:reduce}}\label{proof:hedge_American} 

$(i) \implies (ii):$ Let $\widetilde{H}$ be such that \eqref{eq:hedge} holds true.  We construct a generalized strategy $\widetilde{G} = (G,(G(\theta))_{\theta \in \Theta})$ by $G = H, G(\theta) = H(\theta), \forall \theta \in \Theta$ and get \eqref{eq:hedge2_2}. 

$(ii) \implies (i):$ Let $\widetilde{G} = \left( G, (G(\theta))_{\theta \in \Theta} \right) $ be a strategy satisfying \eqref{eq:hedge2_2}, i.e., there exists $\widetilde{G} \in \widetilde{ \mathcal{A}}^{\Theta}$ such that $\forall
k \in \mathbb{N}, \theta, \zeta \in \Theta,$
\begin{equation}\label{superhedging2}
	z + \int_0^{\zeta \wedge T_k} \widetilde{G}_u(\theta) dS_u - \Phi_{\theta}1_{\theta\le \zeta \wedge T_k }  \ge  - x 1_{\zeta \wedge T_k < \infty}, \ a.s.
\end{equation}
Since $\Phi_t \ge -x, \ a.s.$, we get from \eqref{superhedging2} that $\forall k \in \mathbb{N}, \theta,  \zeta \in \Theta$
\begin{equation*}
\int_0^{\theta} G_u1_{[[0,\zeta\wedge T_k]]}dS_u + \int_{\theta}^{\infty} G_u(\theta)1_{]]\theta, \zeta \wedge T_k]]} dS_u \ge -(2x+z), \ a.s.
\end{equation*}
or equivalently 
\begin{equation}\label{eq:G_theta}
G_u(\theta)1_{]]\theta, \infty[[} \in \A_{2x+z}\left(\theta, \int_0^{\theta} G_udS_u \right), \  \forall \theta \in \Theta.
\end{equation}
We define 
\begin{equation}\label{eq:H}
	H_t:= G_t, \qquad H_t(\theta) = G_t(\theta) \ \text{ for all }\theta \in \Theta.
\end{equation} 
Fixing an arbitrary $\tau \in \mathbb{T} \setminus\Theta$, we will find $H(\tau)1_{]]\tau,\infty[[} \in \mathcal{A}_{2x+z}\left( \tau, \int_{0}^{\tau}H_udS_u \right) $ such that for all $\forall
k \in \mathbb{N}, \zeta \in \Theta,$
\begin{equation}\label{superhedging22}
	z + \int_0^{\zeta \wedge T_k } \widetilde{H}_u(\tau) dS_u - \Phi_{\tau}1_{\tau\le \zeta \wedge T_k }  \ge  - x 1_{\zeta \wedge T_k < \infty}, \ a.s.
\end{equation}
For each $n \in \mathbb{N}$, we  define a new sequence of exercise times
$$\tau_n = (m+1)/2^n \text{ on } A^m_{n}= \left\lbrace m/2^n \le \tau < (m+1)/2^n \right\rbrace \in \mathcal{F}_{(m+1)/2^n}, \ m \in \mathbb{N}$$
and $\tau_n = \infty$ on $\{\tau = \infty\}$. Then $\tau_n \downarrow \tau$. We define 
\begin{eqnarray*}
	H_t(\tau_n) &:=&   \sum_{m \in \mathbb{N}}G_t((m+1)/2^n)1_{A^m_{n}}. 
\end{eqnarray*}
Since $A^m_n \in \mathcal{F}_{(m+1)/2^n}$, the strategy $H(\tau_n)1_{]]\tau_n,\infty[[}$ is predictable. Furthermore, the strategy
$H(\tau_n)1_{]]\tau_n,\infty[[} \in \mathcal{A}_{2x+z}\left(\tau_n,\int_{0}^{\tau_n}H_udS_u \right) $ because of \eqref{eq:G_theta}. Define 
\begin{eqnarray*}
	\widetilde{H}_t(\tau_n) &:=& H_t1_{t \le \tau_n} + H_t(\tau_n)1_{t > \tau_n}. 
%	&=&G_t1_{t \le \tau_n} + 1_{t > \tau_n} \sum_{m \in \mathbb{N}}G((m+1)/2^n)1_{A^m_{n}}
\end{eqnarray*}
From \eqref{superhedging2},  we have $\forall k \in \mathbb{N}, \zeta \in \Theta$
\begin{eqnarray*}
	z + \int_0^{\zeta \wedge T_k} \widetilde{H}_u(\tau_n) dS_u - \Phi_{\tau_n}1_{\tau_n\le \zeta \wedge T_k }  \ge  - x 1_{\zeta \wedge T_k <\infty}, \ a.s.
\end{eqnarray*}
or equivalently,
\begin{eqnarray}
&&\int_0^{\tau} H_u1_{[[0, \zeta \wedge T_k]]} dS_u +  \left( \int_{\tau }^{\tau_n } H_u1_{]]\tau, \zeta \wedge T_k]]} dS_u  \right.  \label{eq:phi} \\
&&\left. +   \int_{\tau_n }^{\infty} \widetilde{H}_u(\tau_n)1_{]]\tau_n, \zeta \wedge T_k]]} dS_u\right)  \ge  \Phi_{\tau_n}1_{\tau_n\le \zeta \wedge T_k }  - x 1_{\zeta \wedge T_k <\infty} - z, \ a.s. \nonumber
\end{eqnarray}
Define $\bv = (v^{k,\tau,\zeta})_{k \in \mathbb{N}, \zeta \in \Theta}$ where 
$$v^{k,\tau,\zeta} =  \int_0^{\tau} H_u1_{[[0, \zeta \wedge T_k]]} dS_u.$$
By \eqref{eq:phi}, the sequence of vectors
$$-2x - z \le \bff_n := \left( \Phi_{\tau_n}1_{\tau_n \le \zeta \wedge T_k}  - x1_{\zeta \wedge T_k <\infty} - z \right)_{k \in \mathbb{N}, \zeta \in \Theta}  \in \mathbf{C}^{\mathbb{N} \times \Theta}\left(\tau ,\bv \right),$$
where $\mathbf{C}^{\mathbb{N} \times \Theta}\left(\tau ,\bv \right)$ is defined in \eqref{eq:c_set}. Noting that $\tau_n \downarrow \tau$, we have 
$$\mathbf{f}_n \to \mathbf{f}:= \left( \Phi_{\tau}1_{\tau \le \zeta \wedge T_k}  - x1_{\zeta \wedge T_k <\infty} - z \right)_{k \in \mathbb{N}, \zeta \in \Theta}, \ a.s. $$
by Assumption \ref{assumption_main} (ii). By Proposition \ref{pro:countable_Fatou}, the set $\mathbf{C}^{\mathbb{N} \times \Theta}\left(\tau,\bv \right)$ is Fatou-closed. Therefore, $\bff \in \mathbf{C}^{\mathbb{N} \times \Theta}\left(\tau,\bv \right) $, that is there exists $H(\tau)$ such that for all $k \in \mathbb{N}, \zeta \in \Theta$,  
\begin{eqnarray*}
z + \int_0^{\tau} H_u1_{[[0, \zeta \wedge T_k]]} dS_u + \int_{\tau}^{\infty}H_u(\tau)1_{]]\tau, \zeta \wedge T_k]]}dS_{u} \ge \Phi_{\tau}1_{\tau \le \zeta \wedge T_k} - x1_{\zeta \wedge T_k < \infty}, \ a.s.
\end{eqnarray*}
and $$H(\tau)1_{]]\tau,\zeta \wedge T_k]]} \in  \mathcal{A}_{2x+z}\left( \tau, \int_{0}^{\tau}H_u1_{[[0, \zeta \wedge T_k]]} dS_u \right).$$
The strategy defined by $\widetilde{H}=(H,(H(\tau))_{\tau \in \mathbb{T}})$ satisfies \eqref{superhedging22}. Sending $k$ to infinity in \eqref{superhedging22}, we get that for any $\tau \in \mathbb{T}, \zeta \in \Theta$,
\begin{equation}\label{hedging4}
	z + \int_0^{\zeta} \widetilde{H}_u(\tau) dS_u - \Phi_{\tau}1_{\tau\le \zeta }  \ge  - x 1_{\zeta  <\infty}, \ a.s.
\end{equation}
For any $t \in [0,\infty)$, applying \eqref{hedging4} to the sequence $\zeta_n \downarrow t$ we get \eqref{eq:hedge} and  the proof is complete. 

\subsection{Proof of Theorem \ref{thm:hedge_American}}\label{subsection:proof_American}
Firstly, we consider the case $\Phi$ is bounded and $\Phi_t \ge -x, a.s., \forall t \ge 0$. Define $$\widetilde{\Phi}^{k,\theta,\zeta}:= \Phi_{\theta}1_{\theta\le  \zeta \wedge T_k}   - x 1_{\zeta \wedge T_k <\infty} \ge -2x, \ a.s.,  \ \forall k \in \mathbb{N}, \theta \in \Theta, \zeta \in \Theta.$$
By Proposition \ref{pro:reduce}, it suffices to  prove that
\begin{eqnarray}
&&\sup_{\bZ\in \mathcal{Z}^{A}}E\left[\sum_{k \in \mathbb{N}, \theta\in \Theta, \zeta \in \Sigma}Z_{}^{k,\theta,\zeta}\widetilde{\Phi}^{k,\theta,\zeta}\right] \nonumber \\
&=& \inf \left\lbrace z:   \exists   \widetilde{H} \in \widetilde{\mathcal{A}}^{\Theta}_{}, z + \int_0^{T_k \wedge \zeta} \widetilde{H}_u(\theta) dS_u \ge \widetilde{\Phi}^{k,\theta,\zeta}, \ a.s., \forall k \in \mathbb{N}, \theta \in \Theta, \zeta \in \Theta \right\rbrace \nonumber\\
&=& \inf \left\lbrace z:   \exists  \widetilde{H} \in \widetilde{\mathcal{A}}^{\Theta}_{2x+z}, 	z + \int_0^{ T_k \wedge \zeta} \widetilde{H}_u(\theta) dS_u \ge \widetilde{\Phi}^{k,\theta, \zeta}, \ a.s., \forall k \in \mathbb{N}, \theta \in \Theta, \zeta \in \Sigma\right\rbrace. \nonumber \\
	\label{eq:super}
\end{eqnarray}
The second quality of \eqref{eq:super} is clear. We prove the first equality of \eqref{eq:super}. For any $\bZ \in \mathcal{Z}^{A}$, we have
\begin{eqnarray}\label{eq11}
	z + E\left[ \sum_{k \in \mathbb{N},\theta \in \Theta, \zeta \in \Theta} Z^{k,\theta, \zeta} \widetilde{W}^{k,\theta, \zeta}_{\infty}(\widetilde{H}) \right] \ge E\left[ \sum_{k \in \mathbb{N},\theta \in \Theta, \zeta \in \Theta} Z^{k,\theta, \zeta} \widetilde{\Phi}^{k,\theta, \zeta}  \right] .
\end{eqnarray}
By Proposition \ref{pro:1}, we have 
\begin{eqnarray*}
	E\left[ \sum_{k \in \mathbb{N}, \theta \in \Theta, \zeta \in \Theta} Z^{k,\theta, \zeta} \widetilde{W}^{k,\theta, \zeta}_{\infty}(\widetilde{H}) \right] \le 0.
\end{eqnarray*}
Therefore, $z \ge \sup_{\bZ \in \mathcal{Z}^{A}} E\left[ \sum_{k \in \mathbb{N}, \theta \in \Theta, \zeta \in \Theta} Z^{k,\theta, \zeta}\widetilde{\Phi}^{k,\theta, \zeta} \right].$ 
Next, we prove the reverse inequality. Let $z \in \mathbb{R}$ be such that there is no strategy $\widetilde{H} \in \widetilde{\mathcal{A}}^{\Theta}_{}$ satisfying
$$ z + \widetilde{W}^{k,\theta, \zeta}_{\infty}(\widetilde{H}) \ge \widetilde{\Phi}^{k,\theta, \zeta}, a.s., \ \forall (k,\theta, \zeta)\in \mathbb{N} \times \Theta \times \Theta.$$
In other words, $(\widetilde{\Phi}^{k,\theta, \zeta})_{(k,\theta, \zeta)\in \mathbb{N} \times \Theta \times \Theta} - z \notin \widetilde{\bC}^{\mathbb{N} \times \Theta \times \Theta}.$ From Theorem \ref{thm:local_closed_0},  the Hahn-Banach theorem implies that there exists $\mathbf{Q} = (Z^{k,\theta, \zeta})_{(k,\theta, \zeta)\in \mathbb{N} \times \Theta \times \Theta}$ in $\bigoplus_{(k,\theta, \zeta)\in \mathbb{N} \times \Theta \times \Theta} L^1(\mathcal{F}_{ T_k \wedge \zeta},P)$ such that 
$$  \sup_{\bff \in \widetilde{\bC}^{\mathbb{N} \times \Theta \times \Theta}} \mathbf{Q}(\bff) \le \alpha < \beta \le \mathbf{Q}\left((\widetilde{\Phi}^{k,\theta, \zeta})_{(k,\theta, \zeta)\in \mathbb{N} \times \Theta \times \Theta} - z \right) .$$ 
Since $\widetilde{\bC}^{\mathbb{N} \times \Theta \times \Theta}$ is a cone containing $-\bL^{\infty}_+$, it is necessarily that
$$ \sup_{\bff \in \widetilde{\bC}^{\mathbb{N} \times \Theta \times \Theta}} \mathbf{Q}(\bff)  = 0, \qquad  \mathbf{Q}\left((\widetilde{\Phi}^{k,\theta, \zeta})_{(k,\theta, \zeta)\in \mathbb{N} \times \Theta \times \Theta} - z \right) > 0.$$
We also deduce that $Z^{k,\theta, \zeta} \ge 0, a.s., (k,\theta, \zeta)\in \mathbb{N} \times \Theta \times \Theta$ and it is possible to normalize $\bQ$ such that $\bQ(\mathbf{1}) = 1$. This means $\bZ \in \mathcal{Z}^{A}$ and that 
$$z < \mathbf{Q}\left((\widetilde{\Phi}^{k,\theta, \zeta})_{(k,\theta, \zeta)\in \mathbb{N} \times \Theta \times \Theta} \right) \le \sup_{\bZ \in \mathcal{Z}^{A}} E\left[ \sum_{(k,\theta, \zeta)\in \mathbb{N} \times \Theta \times \Theta} Z^{k,\theta, \zeta} \widetilde{\Phi}^{k,\theta, \zeta} \right] .$$
Secondly, we consider the case $\Phi_t \ge 0$ for all $t \ge 0$ and prove that
\begin{eqnarray}
&&\sup_{\bZ\in \mathcal{Z}^{A}}E\left[\sum_{k \in \mathbb{N}, \theta\in \Theta, \zeta \in \Theta}Z^{k,\theta, \zeta}\widetilde{\Phi}^{k,\theta, \zeta}\right] \nonumber\\
&=& \inf \left\lbrace z:   \exists  \widetilde{H} \in \widetilde{\mathcal{A}}^{\Theta}_{x+z}, 	z + \int_0^{T_k \wedge \zeta} \widetilde{H}_u(\theta) dS_u \ge \widetilde{\Phi}^{k,\theta, \zeta}, \ a.s., \forall
	k \in \mathbb{N}, \theta \in \Theta, \zeta \in \Theta \right\rbrace. \nonumber \\
	\label{eq:super2}
\end{eqnarray}
The inequality $(\leq)$ in \eqref{eq:super2} follows by the same manner. Consider the reverse inequality. In this case, it holds that  $\widetilde{\Phi}^{k,\theta, \zeta} \ge - x, \ a.s.$ Let $z \in \mathbb{R}$ be such that
$z > \sup_{\bZ \in \mathcal{Z}^{A}} E\left[ \sum_{(k,\theta, \zeta)\in \mathbb{N} \times \Theta \times \Theta} Z^{k,\theta, \zeta} \widetilde{\Phi}^{k,\theta, \zeta} \right].$ Then for all $n \in \mathbb{N}$, we have
$$z > \sup_{\bZ \in \mathcal{Z}^{A}} E\left[ \sum_{(k,\theta, \zeta)\in \mathbb{N} \times \Theta \times \Theta} Z^{k,\theta, \zeta} \left(  \widetilde{\Phi}^{k,\theta, \zeta} \wedge n\right)  \right]. $$
The result for the bounded case implies for each $n \in \mathbb{N}$, there exist $-x \le z_n  \le z$ and  $\widetilde{H}^n \in \widetilde{\mathcal{A}}^{\Theta}_{x + z_n}$ such that 
$$ z_n + \widetilde{W}^{k,\theta, \zeta}_{\infty}(\widetilde{H}^n) \ge  \widetilde{\Phi}^{k,\theta, \zeta} \wedge n, \ a.s., \ \forall (k,\theta, \zeta) \in \mathbb{N} \times \Theta \times \Theta.$$
Up to a subsequence, we may assume that $z_n \to z^* \le z$. The strategy $\widetilde{H}^n \in \widetilde{\mathcal{A}}^{\Theta}_{x+z}$ for each $n \in \mathbb{N}$. The sequence $\widetilde{\bW}_n = \left(\widetilde{W}^{k,\theta, \zeta}_{\infty}(\widetilde{H}^n) \right)_{(k,\theta, \zeta) \in \mathbb{N} \times \Theta \times \Theta} $ is in $\widetilde{\bK}^{\mathbb{N} \times \Theta \times \Theta}_{x+z}$, which is $c$-bounded by Lemma \ref{lemma:c-bounded}. We may take convex combination $\widehat{\bW}_n \in  conv(\widetilde{\bW}_n, \widetilde{\bW}_{n+1},...)$ which converges almost surely to a vector $\widehat{\bW}$. We still have $z^* + \widehat{\bW} \ge (\widetilde{\Phi}^{k,\theta, \zeta})_{(k,\theta, \zeta) \in \mathbb{N} \times \Theta \times \Theta}$. Since $\widetilde{\bC}^{\mathbb{N} \times \Theta \times \Theta}_0$ is Fatou-closed by Proposition \ref{pro:countable_Fatou_A}, there exists a strategy $\widetilde{H}^* \in \widetilde{\mathcal{A}}^{\Theta}_{x+z}$ such that 
$$z^* + \widetilde{W}^{k,\theta, \zeta}_{\infty}(\widetilde{H}^*) \ge \widetilde{\Phi}^{k,\theta, \zeta}, \ a.s., \ \forall (k,\theta, \zeta) \in \mathbb{N} \times \Theta \times \Theta.$$ 
Hence, we get 
$$z \ge \inf\left\lbrace  v \in \mathbb{R}: \exists \widetilde{H} \in \widetilde{\mathcal{A}}^{\Theta}_{x+v}, v + \widetilde{W}^{k,\theta, \zeta}_{\infty}(\widetilde{H}^*) \ge \widetilde{\Phi}^{k,\theta, \zeta}, \ a.s., \forall
(k,\theta, \zeta) \in \mathbb{N} \times \Theta \times \Theta
\right\rbrace.$$ The proof is complete.
\subsection{Proof of Proposition \ref{pro:appro_hedging}}
Denote by $\mathcal{Q}^k$ the set of Radon-Nikodym densities of  equivalent local martingale measures for the market up to time $T \wedge T_k$. By the classical superhedging duality, we compute
	\begin{eqnarray*}
		\pi_{k}(G_{T},x) &=& \sup_{Z \in \mathcal{Q}^k} E\left[ Z_{T_k \wedge T}\left( G_{T}1_{T \le  T_k} - x1_{ T_k <T}\right) \right] \\
		&=& \sup_{Z \in \mathcal{Q}^k} E\left[ Z_{T_k \wedge T}\left( G_{T}1_{T \le  T_k} + x1_{T_k \ge T}\right) \right]  - x\\
		&=&\sup_{Y \in ELMD} E\left[ Y_{T_k \wedge T}\left( G_{T} + x\right)1_{ T_k \ge T} \right]  - x\\
		&=&\sup_{Y \in ELMD} E\left[ Y_{T}\left( G_{T} + x\right)1_{ T_k \ge T} \right]  - x
	\end{eqnarray*}
	where the third equality is due to the fact that for each density $Z_{T \wedge T_k}$, there exists an ELMD $Y$ such that $Z_{T\wedge T_k}1_{T_k \ge T} = Y_{T \wedge T_k}1_{T_k \ge T}$ and vice versa.  
	
	The sequence $\pi_k(G_{T},x), k \in \mathbb{N}$ is increasing, and bounded by $\pi_T(G_T,x)$.  Indeed, for each $k \in \mathbb{N}$, let $H^k$ be a superhedging strategy corresponding to the price $\pi_k(G_{T},x)$ for the payoff $G_{T}1_{T \le T_k } - x1_{T_k <T}$  when trading until $T \wedge T_k $, i.e.,
	\begin{equation}\label{eq:super_k}
		\pi_k(G_{T},x) + \int_0^{T_k \wedge T} H^k_u dS_u \ge  G_{T}1_{ T \le T_k} - x1_{T_k <   T}, \ a.s.
	\end{equation}
	This comes from the classical superhedging duality. For the price $\pi_T(G_T,x)$, it satisfies the condition \eqref{eq:super_Eur_prod_2},  so it is clear that $\pi_k(G_{T},x) \le  \pi_T(G_T,x)$ for all $k \in \mathbb{N}$.  Therefore, there exists the limit $\pi^* = \lim_{k \to \infty}\pi_k(G_{T},x) \le \pi_T(G_{T},x)$. 
	
	We prove that $\pi^* \ge \pi_T(G_{T},x)$. From \eqref{eq:super_k},   extending $H^k = 0$ on $]]T \wedge T_k,T]]$, we obtain for all $n \ge k$ that
	\begin{equation*}
		\pi_k(G_{T},x) + \int_0^{T_n \wedge T} H^k_u dS_u \ge G_{T}1_{ T \le T_k } - x1_{T_k < T}, \ a.s.
	\end{equation*}
	For $n < k$, it holds that
	\begin{equation*}
		\pi_k(G_{T},x) + \int_0^{T_n \wedge T} H^k_u dS_u \ge  G_{T}1_{T \le  T_n } - x1_{T_n < T \le \ T_k } -  x1_{ T_k <T}, \ a.s.
	\end{equation*}
	This means that for all $k \in \mathbb{N}$, the vector $\bff_k =(f^n_k)_{n\in \mathbb{N}}$ defined by
	$$ f^n_k=
	\left\{
	\begin{array}{ll}
		G_{T}1_{ T \le T_k } - x1_{T_k < T} - \pi_k(G_{T},x)  & \mbox{if } n \geq k \\
		G_{T}1_{T \le  T_n } - x1_{T_n < T \le \ T_k } -  x1_{ T_k <T} - \pi_k(G_{T},x)& \mbox{if } n < k
	\end{array} 
	\right., $$
	is in $\mathbf{C}^{\mathbb{N}  \times \{T\}}(0,\mathbf{0})$, see in  \eqref{eq:c_set}. 	Noting that $\bff_k \ge -2x - \pi^*, \ a.s.$ for all $k \in \mathbb{N}$ and 
	$$\bff_k \to \left( G_{T}1_{T \le T_n} - x 1_{T_n <T} - \pi^* \right)_{n \in \mathbb{N}}$$
	as $k \to \infty$. The vector $\left( G_{T}1_{T \le T_n} - x 1_{T_n <T}  -  \pi^* \right)_{n \in \mathbb{N}}$ is in $\mathbf{C}^{\mathbb{N} \times \{T\}}(0,\mathbf{0})$ by Proposition \ref{pro:countable_Fatou}, and thus there exists a strategy $H$  such that for all $n \in \mathbb{N}$, we have $H1_{]]0,T_n \wedge T]]} \in \A_{2x+\pi^*}(0,0)$ and 
	$$\pi^* + \int_0^{T_n \wedge T} H_u dS_u - G_{T}1_{T\le  T_n }  \ge  - x 1_{T_n <T}, \ a.s,$$
	Therefore, $\pi^* \ge \pi_T(G_T,x).$ Finally, the dominated convergence theorem yields
	\begin{eqnarray*}
		\lim_{k \to \infty} \sup_{Y \in ELMD}  E\left[ Y_{T }\left( G_{T} +x\right) 1_{T \le T_n }  \right] =  \sup_{Y \in ELMD} E\left[ Y_{T}(G+x)\right]
	\end{eqnarray*}
	and then \eqref{eq:superheding} follows. 
\section{Appendix}\label{sec:app}
\subsection{Convex compactness and maximal elements in $(L^0_+)^{\mathbb{N}}$}\label{sec:convex_comp}

%Let $L^{\S}$ denote the set of 
%$[0,\infty]$-valued random variables, equipped with the topology of convergence in probability.
A set $A\subset L^0_{+}$ is \emph{bounded} if $\sup_{X\in A}P(X\geq n)\to 0$, $n\to\infty$.
Now consider the topological product $\mathbf{L}:=(L^0_+)^{\mathbb{N}}$.
We call a subset $C\subset\mathbf{L}$ \emph{c-bounded}, if $\pi_k(C)$
is bounded in $L^0_{+}$ for all coordinate mappings $\pi_k:\mathbf{L}\to L^{0}_{+}$, $k\in\mathbb{N}$.

For any set $A$ we denote by $\mathrm{Fin}(A)$ the family of all non-empty finite subsets of $A$. This
is a directed set with respect to the partial order induced by inclusion. 
We recall  Definition 2.1 of \cite{gordan}.

\begin{definition}
	{\rm A convex subset $C$ of some topological vector space is \emph{convexly compact},
		if for any non-empty set $A$ and any family $F_a$, $a\in A$ of closed and convex
		subsets of $C$, one has $\cap_{a\in A}F_a\neq\emptyset$ whenever
		\[
		\forall B\in\mathrm{Fin}(A),\ \cap_{a\in B}F_a\neq\emptyset.
		\]}
\end{definition}

It was established, independently in both \cite{pratelli} and \cite{gordan}, that every
closed and bounded convex subset of $L_+^0$ is convexly compact. The following criterion
was formulated in Proposition 4.2 of \cite{cfr}. 

\begin{proposition}\label{bfl}
	Any c-bounded, convex and closed subset $C\subset\mathbf{L}$ is convexly compact.
\end{proposition}

For $\mathbf{f}=(f_{0},f_{1},\ldots),\mathbf{g}=(g_{0},g_{1},\ldots)\in\mathbf{L}$, we write
$$
\mathbf{f}\preceq \mathbf{g}
$$ 
when $f_{k}\leq g_{k}$ for all $k\in\mathbb{N}$.

\begin{lemma}\label{lemma_maximal} Each (nonempty) closed $c$-bounded set $C\subset \mathbf{L}$ contains a maximal element
	with respect to the partial ordering $\preceq$.{}
\end{lemma}
\begin{proof}
	We apply transfinite recursion along the ordinals $0\leq \alpha<\omega_{1}$, where $\omega_{1}$ is the
	first uncountable cardinal.	
	
	Let $\mathbf{f}^{0}\in C$ be arbitrary. If $\alpha=\beta +1$ is a successor ordinal and $\mathbf{f}^{\beta}$ is not
	maximal then we can find $\mathbf{f}^{\alpha}=\mathbf{f}^{\beta+1}\in C$ such that $\mathbf{f}^{\beta}_{k}\leq{}
	\mathbf{f}^{\alpha}_{k}$ for each $k\in\mathbb{N}$ and there exists $k_{0}\in\mathbb{N}$ such that
	$P(\mathbf{f}^{\beta}_{k_{0}}<\mathbf{f}^{\alpha}_{k_{0}})>0$. If $\mathbf{f}^{\beta}$ is maximal
	then set $\mathbf{f}^{\alpha}:=\mathbf{f}^{\beta}$. 
	
	If $\alpha$ is a limit ordinal, we take a sequence $\alpha_{n}$, $n\in\mathbb{N}$ cofinal in $\alpha$ and
	define $\mathbf{f}^{\alpha}_{k}:=\sup_{n\in\mathbb{N}}\mathbf{f}^{\alpha_{n}}_{k}$, $k\in\mathbb{N}$.
	We claim that $\mathbf{f}^{\alpha}_{k}$ is a.s.\ finite for each $k$. Indeed, $c$-boundedness of $C$ implies
	that $\mathbf{f}^{\alpha_{n}}_{k}$, $n\in\mathbb{N}$ is bounded in $L^{0}_{+}$ hence its limit is necessarily
	a.s.\ finite. Now the closedness of $C$ implies that $\mathbf{f}^{\alpha}\in C$.
	
	The above construction gives us a family $\mathbf{f}^{\gamma}$, $\gamma<\omega_{1}$. Define
	$$
	\Lambda({\gamma}):=\sum_{k=0}^{\infty} 2^{-k}E[e^{-\mathbf{f}^{\gamma}_{k}}]
	$$
	and let $\Lambda_{*}:=\inf_{\gamma}\Lambda(\gamma)$. As every point in $\mathbb{R}$ has a countable
	neighbourhood basis, there is a sequence $\gamma_{n}$ such that $\Lambda(\gamma_{n})$ tends
	to $\Lambda_{*}$ in a non-increasing way. The ordinal $\gamma_{*}:=\sup_{n}\gamma_{n}$ is necessarily
	countable, that is $\gamma_{*}<\omega_{1}$. But this means that $$
	\Lambda_{*}\leq\Lambda(\gamma_{*})\leq \inf_{\gamma}\Lambda(\gamma)=\Lambda_{*}
	$$
	which guarantees that $\Lambda(\gamma_{*})=\Lambda_{*}$ hence
	$\mathbf{f}^{\gamma}_{k}= \mathbf{f}^{\gamma_{*}}_{k}$ for all $\gamma\geq\gamma_{*}$
	and for all $k\in\mathbb{N}$. But this means that $\mathbf{f}_{\gamma_{*}}$ is a maximal element for $\preceq$. \qed  
\end{proof}

\subsection{Some topological spaces}
We recall some topological vector spaces for convenience of the readers.

\textbf{Inductive topologies.} Let $X$ be a vector space and $I$ be an index set such that for each 
$i \in I$, we are given a locally convex space $(X_i, \tau_i)$ and a linear mapping $f_i: X_i \to X$. 
The inductive topology of the family of spaces $(X_i, \tau_i)$ with respect to the family of 
mappings $f_i$ is the strongest locally convex topology in $X$ such that all mappings $f_i$ are continuous. 
The inductive limit of the family $X_i$ with respect to the mappings $f_i$ is the vector space $X$ 
equipped with this topology.  

\textbf{Topological direct sums.}
Let $I$ be a non-empty set and, for each $i \in I$, let $(X_i,\tau_i)$ be a locally convex topological spaces. 
The topological direct sum of the family $(X_i,\tau_i)$, denoted by 
$\bigoplus_{i \in I} (X_i, \tau_i)$, is the locally convex space defined as follows. 
The vector space $\bigoplus_{i \in I} X_i$ is the set of tuples $(x_i)_{i \in I}$ with $x_i \in X_i$ such that 
$x_i =0$ for all but finitely many $i$. It is equipped with the inductive topology with respect to the canonical embeddings 
\begin{eqnarray*}
	e_i: (X_i,\tau_i) &\to& X\\
	x_i &\mapsto& x = (x^i),
\end{eqnarray*}
where $x^i = x_i$ and $x^j = 0$ whenever $j \ne i$, i.e.\ the 
strongest locally convex topology on $\bigoplus_{i \in I} X_i$ such that all these embeddings are continuous.

\textbf{LF space (countable inductive limits of Fr\'echet spaces)} Let $X_n , n \in \mathbb{N}$ be an increasing sequence of linear
subspaces of a vector space $X$, i.e. $X_n \subset X_{n+1}$ for all $n \in N$, such that
$X = \bigcup_{n \in \mathbb{N}}X_n$. For each $n \in \mathbb{N}$,  let $(X_n, \tau_n)$ be a Fr\'echet space such that the
natural embedding in of $X_n$ into $X_{n+1}$ is a topological isomorphism, i.e. the
topology induced by $\tau_{n+1}$ on $X_n$ coincide with $\tau_n$. The space $X$ equipped with 
the inductive topology $\tau$ w.r.t. the family $(X_n, \tau_n), n \in \mathbb{N}$ is said to be 
the LF-space with defining sequence $(X_n, \tau_n), n \in \mathbb{N}$.

The direct sum of a finite number of Fréchet spaces is a Fréchet space. 
It could be checked that the direct sums of a countable sequence of Fr\'echet spaces is an LF space. 
An LF space is not necessarily metrizable, see Section 27.41 of \cite{Schechter96}. 
For duality between direct sums and product spaces, we refer to \cite{Schechter96},  
\cite{bogachev2017topological},  \cite{mv97} and Section 4 of \cite{cfr}. 

We define $B^{\infty}_{x}= \{ f \in L^{\infty}(\mathcal{F}_T,P): \|f\|_{\infty} \le x\}$, the closed ball of radius $x\geq 0$ in 
$L^{\infty}(\mathcal{F}_T,P)$. The following result is given in Proposition 4.1 of \cite{cfr}.   
\begin{proposition}\label{cfr}
	Let $D$ be a finite index set and $\bC \subset \mathbf{L}^{\infty,D}$ be a convex set. The set $\bC$ is closed in the $\bw^*$ topology if and only if 
	$\bC \cap \prod_{k \in D} B^{\infty}_{x}$ is closed in $ \mathbf{L}^{0,D}(\mathcal{F}_T,P)$
	for each $x\geq 0$. 
\end{proposition}




\end{document}